%% file: PRAMsLB.tex
\documentclass[conference]{IEEEtran}
\usepackage[utf8]{inputenc}
\usepackage[T1]{fontenc}
\usepackage{amsthm,amssymb,cmll,bussproofs,stmaryrd,mathtools,tikz}
\usepackage{tensor}
\usepackage{multicol}
\usepackage{environ}
\usepackage{tikz}
\usepackage{pgfplots}
\usetikzlibrary{matrix,arrows,calc,shapes}
\usepackage{microtype}
\usepackage{hyperref}
\usepackage{cleveref}
\usepackage{color}

\usepackage{fullmacros-thomas}
\usepackage[inline]{asymptote}

\newcommand{\pointmedian}{\fontfamily{cmr}\selectfont\textperiodcentered}

\newcommand{\cuckersproblem}{\mathfrak{F}\mathrm{er}}

\newcommand{\gct}{{\sc gct}\xspace}

\newcommand{\seq}[1]{\mathbf{#1}}

\newcommand{\inputspaceE}[1]{\Space{In}^{\mathrm{E}}(#1)}
\newcommand{\inputspaceI}[1]{\Space{In}^{\mathrm{I}}(#1)}

\newcommand{\inputs}[1]{\mathrm{In}(#1)}
\newcommand{\outputs}[1]{\mathrm{Out}(#1)}

\newcommand{\history}[1]{\mathfrak{h}(#1)}
\newcommand{\benorvar}[2][\seq{e'}]{x^{\history{#1}(#2)}_{#2}}

\newcommand{\cotree}[2][T]{\textsc{coT}_{#2}(#1)}

\newcommand{\maxflow}{\ensuremath{\mathtt{maxflow}}\xspace}

\newcommand{\prams}{{\sc pram}{\rm s}\xspace}
\newcommand{\pram}{{\sc pram}\xspace}
\newcommand{\acts}{{\sc act}{\rm s}\xspace}

\newcommand{\algcirc}{{\sc algcirc}\xspace}
\newcommand{\srams}{{\sc ram}{\rm s}\xspace}
\newcommand{\sram}{{\sc ram}\xspace}

\newcommand{\crew}{\textsc{crew}\xspace}

\newcommand{\instruction}[1]{\mathtt{#1}}

\newcommand{\copyy}[2]{\instruction{copy}(#1,#2)}
\newcommand{\copyref}[2]{\instruction{copy}(#1,\sharp{}#2)}
\newcommand{\refcopy}[2]{\instruction{copy}(\sharp{}#1,#2)}

\newcommand{\boadd}[3]{\instruction{+}_{#1}(#2,#3)}
\newcommand{\bosubstract}[3]{\instruction{-}_{#1}(#2,#3)}
\newcommand{\bomultiply}[3]{\instruction{\times}_{#1}(#2,#3)}
\newcommand{\boaddconst}[3]{\instruction{+}^{#3}_{#1}(#2)}
\newcommand{\bosubstractconst}[3]{\instruction{-}^{#3}_{#1}(#2)}
\newcommand{\bomultiplyconst}[3]{\instruction{\times}^{#3}_{#1}(#2)}
\newcommand{\bodivide}[3]{\instruction{/}_{#1}(#2,#3)}
\newcommand{\bodivideconst}[3]{\instruction{/}^{#3}_{#1}(#2)}
\newcommand{\bogen}[3]{\star_{#1}(#2,#3)}
\newcommand{\bogenconst}[3]{\star^{#3}_{#1}(#2)}
\newcommand{\boeuclidivide}[3]{\instruction{//}_{#1}(#2,#3)}
\newcommand{\boconst}[2]{\instruction{const}_{#1}(#2)}

\newcommand{\bosqrtn}[2]{\instruction{\sqrt[n]{#1}}(#2)}
\newcommand{\bosqrt}[3][n]{\instruction{\sqrt[#1]{#2}}(#3)}

\newcommand{\rootdegree}[1]{\sqrt[\partial]{#1}}

\newcommand{\plusplus}{\textrm{\scriptsize{$+\!+$}}}
\newcommand{\simref}{\sim_{\mathrm{ref}}}

\newcommand{\amcact}{\alpha_{\mathrm{act}}}
\newcommand{\amcram}{\alpha_{\mathrm{ram}}}
\newcommand{\amcprams}{\alpha_{\mathrm{pram}}}

\newcommand{\amcfull}{\alpha_{\mathrm{full}}}
\newcommand{\amcrealfull}{\alpha_{\realN\mathrm{full}}}

\newcommand{\seqedges}[2]{\mathrm{Seq}_{#1}(#2)}

\newcommand{\command}[1]{\mathtt{#1}}
\newcommand{\instr}[2][M]{\mathrm{Inst}_{#1}(#2)}
\newcommand{\ncproduct}[2]{\tensor*[_{#1}]{\ast}{_{#2}}}

\newcommand{\admss}[2][G]{\mathrm{Adm}_{#2}(#1)}

\newcommand{\hyperplan}[1]{\mathbb{H}(#1)}
\newcommand{\complementset}[1]{#1^{\mathrm{c}}}

\newcommand{\conditional}[3]{\texttt{if}~ {\command{#1}}
  ~ \texttt{goto} ~ {#2} ~ \texttt{else} ~ {#3}}
\newcommand{\Skip}{\mathtt{skip}}
\newcommand{\length}[1]{\left| #1 \right|}

\newcommand{\opencovers}[1]{\mathrm{Cov}(\Space{#1})}
\newcommand{\finiteopencovers}[1]{\mathrm{FCov}(\Space{#1})}

\newcommand{\OptProb}{\mathcal{P}_{\mathrm{opt}}}
\newcommand{\DecProb}{\mathcal{P}_{\mathrm{dec}}}
\newcommand{\MaxOptProb}{\mathrm{Max}\OptProb}
\newcommand{\Parametrization}{\mathsf{Param}}
\newcommand{\AffinePlane}{\mathcal{A}_1}

\newcommand{\Frontier}{\mathsf{Front}}
\newcommand{\Fan}{\mathsf{Fan}}

\newcommand{\projectionAz}[1]{\Pi(#1)}

\newcommand{\Interpret}[1]{\mathopen{|\![} #1 \mathclose{]\!|}}

\ifCLASSINFOpdf
\else
\fi
\begin{document}
%
\title{Lower bounds for algebraic machines, semantically}

\author{\IEEEauthorblockN{Luc Pellissier}
\IEEEauthorblockA{LACL, Faculté des Sciences et Technologie\\
61 avenue du Général de Gaulle\\
94010 Créteil, FRANCE\\
Email: luc.pellissier@lacl.fr}
\and
\IEEEauthorblockN{Thomas Seiller}
\IEEEauthorblockA{CNRS, LIPN -- UMR 7030\\ 
Université Sorbonne Paris Nord\\
99, Avenue Jean-Baptiste Clément\\
93430, Villetaneuse, FRANCE\\
Email: seiller@lipn.fr}}


%


\maketitle

\begin{abstract}
  This paper presents a new semantic method for proving lower bounds in
  computational complexity. We use it to prove that \maxflow, a \Ptime complete 
  problem, is not computable in polylogarithmic time on parallel random access 
  machines (\prams) working with integers, showing that $\NCInteger\neq\Ptime$, 
  where \NCInteger is the complexity class defined by such machines, 
  and \Ptime is the standard class of 
  polynomial time computable problems (on, say, a Turing machine). 
  On top of showing this new separation result, we show our method  
  captures previous lower bounds results from the literature:
  Steele and Yao's lower bounds for algebraic decision trees \cite{SteeleYao82}, 
  Ben-Or's lower bounds for algebraic computation trees \cite{Ben-Or83}, Cucker's
  proof that \NCReal is not equal to \PtimeReal \cite{Cucker92}, and Mulmuley's
  lower bounds for \enquote{\prams without bit operations} \cite{Mulmuley99}.
\end{abstract}

\input{introduction.tex}

\input{models-programs.tex}

\input{entropy}

\input{benor}

\input{surfaces-optimization}

\input{result}

\bibliographystyle{IEEEtran/IEEEtran}
\bibliography{biblio}

\appendix

\input{omitted.tex}

\end{document}

%% file: introduction.tex
\section{Introduction}

Complexity theory has traditionally been concerned with proving 
\emph{separation results} between complexity classes. Many problems
remain open, such as the much advertised \Ptime vs \NPtime question
which concerns the difference between feasible sequential computability 
in deterministic and non-deterministic models, or equivalently the difference
between feasible computation and feasible verification.
This paper investigates questions related to the \NC vs \Ptime question,
which concerns the difference between efficient sequential computation
and (more than) efficient parallel computation. 

Proving that two classes $B\subset A$ are not equal can be reduced to 
finding lower bounds for problems in $A$: by proving that certain problems 
cannot be solved with less than certain resources on a specific model of 
computation, one can show that two classes are not equal. Conversely, 
proving a separation result $B\subsetneq A$ provides a lower bound for 
the problems that are $A$\emph{-complete} \cite{cooknpcomplete} -- i.e. 
problems that are in some way \emph{universal} for the class $A$.

The proven lower bound results are however very rough, and many separation problems
remain as generally accepted conjectures. For instance, a proof that the class
of non-deterministic exponential problems is not included in what is thought of
as a very small class of circuits was not achieved until very recently
\cite{Williams}.
%

The failure of most techniques of proof has been studied in itself, which lead
to the proof of the existence of negative results that are commonly called
\emph{barriers}. Altogether, these results show that all proof methods we know
are ineffective with respect to proving interesting lower bounds. Indeed, there
are three barriers: relativisation \cite{relativization}, natural proofs
\cite{naturality} and algebrization \cite{algebraization}, and almost every
known proof method hits at least one of them: this shows the need for new
methods\footnote{In the words of S. Aaronson and A. Wigderson
  \cite{algebraization}, \enquote{We speculate that going beyond this limit
    [algebrization] will require fundamentally new methods.}}. However, to this
day, only one research program aimed at proving new separation results is
commonly believed to have the ability to bypass all barriers: Mulmuley and
Sohoni's Geometric Complexity Theory (\gct) program \cite{GCTsurvey2}. This
research program was inspired from an earlier lower bounds result by Mulmuley
\cite{Mulmuley99} which we strengthen in this paper.

\subsection{Mulmuley's result}

The \NC vs \Ptime question is one of the foremost open questions in
computational complexity. In laymen's terms, it asks whether a problem
efficiently computable on a sequential machine can be computed substantially
more efficiently on a parallel machine. It is well known that any problem in
\NC, i.e. that is computable in polylogarithmic time on a parallel machine (with a
polynomial number of processors), belongs to \Ptime, i.e. is computable in
polynomial time on a sequential machine. The converse, however, is expected to
be false. Indeed, although many problems in \Ptime can be shown to be in \NC,
some of them seem to resist efficient parallelisation. In particular it is not
known whether the \maxflow problem, known to be \Ptime-complete
\cite{MaxflowComplete}, belongs to \NC.

As part of the investigations on the \NC vs \Ptime question, a big step forward
is due to K. Mulmuley. In 1999 \cite{Mulmuley99}, he showed that a notion of
machine introduced under the name ``\prams without bit operations'' does not
compute \maxflow in polylogarithmic time. This notion of machine, quite exotic
at first sight, corresponds to an algebraic variant of \prams, where registers
contain integers and individual processors are allowed to perform sums,
subtractions and products of integers. It is argued by Mulmuley that this notion
of machine provides an expressive model of computation, able to compute some non
trivial problems in \NC such as Neff's algorithm for computing approximate roots
of polynomials \cite{Neffsalgo}. Although Mulmuley's result has represented a
big step forward in the quest for a proof that \Ptime and \NC are not equal, the
result was not strenghtened or reused in the last 20 years, and remained the
strongest known lower bound result.

\subsection{Contributions.}

The main contribution of this work is a strengthening of Mulmuley's lower
bounds result. While the latter proves that \maxflow is not computable in
polylogarithmic time in the model of \enquote{\prams without bit operations}, 
we show here that \maxflow is not computable in polylogarithmic time in the more 
expressive model of \prams over integers, making an additional step in the 
direction of a potential proof that \NC is different from \Ptime. Indeed, our
result can be stated as
\begin{theorem}\label{mainthm}\label{Theorem1}
\[\NCInteger \neq \Ptime,\]
where \NCInteger is the set of problems decidable in polylogarithmic 
time by a (not necessarily uniform) family of \prams over $\integerN$.
\end{theorem}

The second contribution of the paper is the proof method itself, which is 
based on \emph{dynamic semantics} for programs by means of \emph{graphings},
a notion introduced in ergodic theory and recently used to define models 
of linear logic by Seiller 
\cite{seiller-igf,seiller-igg,seiller-ignda,seiller-ige}. The dual nature of
graphings, both continuous and discrete, is essential in the present work, 
as it enables invariants from continuous mathematics, in particular the 
notion of \emph{topological entropy} for dynamical systems, while the 
finite representability of graphings is used in the key lemma (as the number
of \emph{edges} appears in the upper bounds of \Cref{thm:graphingsBenOrsystems}).

In particular, we show how this proof method captures
known lower bounds and separation results in algebraic models of computation,
namely Steele and Yao's lower bounds for algebraic decision trees \cite{SteeleYao82},
Ben-Or's lower bounds on algebraic computation trees \cite{Ben-Or83}, 
Cucker's proof that \NCReal is not equal to \PtimeReal (i.e. answering the 
\NC vs \Ptime problem for computation over the real numbers).

\subsection{A more detailed view of the proof method}

One of the key ingredients in the proof is the representation of 
programs as graphings, and \emph{quantitative soundness} results. 
We refer to the next section for a formal statement, and we only provide
an intuitive explanation for the moment. Since 
a program $P$ is represented as a graphing $\Interpret{P}$, which is in 
some way a dynamical system, the computation $P(a)$ on a given 
input $a$ is represented as a sequence of values 
$\Interpret{a}), \Interpret{P}(\Interpret{a}), \Interpret{P}^2(\Interpret{a}),\dots$. 
Quantitative soundness
states that not only $\Interpret{P}$ computes exactly as $P$, but it does
so in the same number of steps, i.e. if $P(a)$ terminates on a value $b$
in time $k$, then $\Interpret{P}^{k}(\Interpret{a})=\Interpret{b}$.

The second ingredient is the dual nature of graphings, both continuous 
and discrete objects. Indeed, a graphing \emph{representative} is a 
graph-like structure whose edges are represented as continuous 
maps, i.e. a finite representation of a (partial) continuous dynamical 
system. Given a graphing, we define its \emph{$k$th 
cell decomposition}, which separates the input space into cells such that 
two inputs in the same cell are indistinguishable in $k$ steps, i.e. the 
graphing's computational traces on both inputs are equal. We can then
use both the finiteness of the graphing representatives and the 
\emph{topological entropy} of the associated dynamical system to 
provide upper bounds on the size of a further refinement of this 
geometric object, namely the \emph{$k$-th entropic co-tree}
of a graphing -- a kind of final approximation of the graphing by a 
computational tree\footnote{Intuitively, the $k$-th entropic co-tree 
mimicks the behaviour of the graphing for $k$ steps of computation.} 

As we deal with algebraic models of computation, this implies a bound 
on the representation of the  $k$th cell decomposition as a semi-algebraic 
variety. In other words, the $k$th cell decomposition is defined by polynomial 
in\pointmedian{}equalities and we provide bounds on the number and 
degree of the involved polynomials. The corresponding statement is the
main technical result of this paper 
(\Cref{thm:graphingsBenOrsystems}).

This lemma can then be used to obtain lower bounds results. Using 
the Milnor-Thom theorem to bound the number of connected components 
of the $k$th cell decomposition, we then recover the lower bounds of 
Steele and Yao on algebraic decision trees, and the refined result of 
Ben-Or providing lower bounds for algebraic computation trees. A different
argument based on invariant polynomials provides a proof of Cucker's 
result that $\NCReal\neq\PtimeReal$ by showing that a given polynomial 
that belongs to $\PtimeReal$ cannot be computed within $\NCReal$. Lastly,
following Mulmuley's geometric representation of the \maxflow problem,
we are able to strenghten his celebrated result to obtain lower bounds on
the size (depth) of a \pram over the integers computing this problem. This
proves the following theorem, which has \autoref{mainthm} as a corollary.

\begin{theorem}
  \label{theorem:main-pram}
  Let $c$ be a positive integer, $M$ a \pram over $\integerN$ with 
  $2^{O((\log N)^c)}$ processors, with $N$ the length of the inputs. 
  Then $M$ does not decide \maxflow in $O((\log N)^c)$ steps.
\end{theorem}

%% file: models-programs.tex
\section{Programs as Dynamical systems}

\subsection{Abstract models of computation and graphings}
\label{sec:amc-graphings}

We consider computations as dynamical processes, hence model them as a dynamical
system with two main components: a space \(\Space{X}\) that abstracts the notion
of configuration space and a monoid acting on this space that represents the
different operations allowed in the model of computation. Although the notion of
\emph{space} considered can vary (one could consider e.g. topological spaces,
measure spaces, topological vector spaces), we restrict ourselves to topological
spaces in this work.

\begin{definition}
  An \emph{abstract model of computation} (\amc) is a monoid action 
  $\alpha: M\acton\Space{X}$, i.e. a monoid morphism from $M$ to the group of
  endomorphisms of $\Space{X}$.
  The monoid $M$ is often given by a set \(G\) of generators and a set of
  relations \(\relations{R}\). 
 We denote such an \amc as $\AMC[X]{G}{R}{\alpha}$.
\end{definition}

Programs in an \amc $\AMC[X]{G}{R}{\alpha}$ is then defined as \emph{graphings},
i.e. graphs whose vertices are subspaces of the space \(\Space{X}\)
(representing sets of configurations on which the program act in the same way)
and edges are labelled by elements of \(\MonGaR{G}{R}\), together with a global
control state. More precisely, we use here the notion of \emph{topological
  graphings}\footnote{While \enquote{measured} graphings were already considered
  \cite{seiller-igg}, the definition adapts in a straightforward manner to allow for
  other notions such as graphings over topological vector spaces -- which would
  be objects akin to the notion of quiver used in representation theory.} 
  \cite{seiller-igg}.
  
\begin{definition}
  An $\alpha$-graphing representative $G$ w.r.t. a monoid action $\alpha: M\acton\Space{X}$ 
  is defined as a set of \emph{edges} $E^{G}$ together with a map that assigns to
  each element $e\in E^{G}$ a pair $(S^{G}_{e},m^{G}_{e})$ of a subspace 
  $S^{G}_{e}$ of $\Space{X}$ -- the \emph{source} of $e$ -- and an element 
  $m^{G}_{e}\in M$ -- the \emph{realiser} of $e$.
\end{definition}

While graphing representatives are convenient to manipulate, they do provide too
much information about the programs. Indeed, if one is to study programs as
dynamical systems, the focus should be on the \emph{dynamics}, i.e. on how the
object acts on the underlying space. The following notion of \emph{refinement} 
captures this idea that the same dynamics may have different graph-like 
representations.

\begin{definition}[Refinement]
  An $\alpha$-graphing representative $F$ is a refinement of an $\alpha$-graphing 
  representative $G$, noted $F\leqslant G$, if there exists a partition 
  $(E^{F}_{e})_{e\in E^{G}}$ of $E^{F}$ such that $\forall e\in E^{G}$:
  $$\begin{array}{c}
    \left(\cup_{f\in E^{F}_{e}} S^{F}_{f} \right) \mathbin{\triangle} S^{G}_{e}
    = \emptyset;
    \hspace{0.5cm}
    \forall f\neq f'\in E^{F}_{e},~ S^{F}_{f} \mathbin{\triangle} S^{F}_{f'}
    = \emptyset;\\
    \hspace{0.5cm}
    \forall f\in E^{F}_{e},~ m^{F}_{f}=m_{e}^{G}.
    \end{array}
  $$
  This induces an equivalence relation defined as 
  \[ F\simref G \Leftrightarrow 
  \exists H, ~ H\leqslant F \wedge H\leqslant G.\]
\end{definition}

The notion of \emph{graphing} is therefore obtained by considering the quotient 
of  the set of graphing representatives w.r.t. $\simref$. Intuitively, this corresponds
to identifying graphings whose \emph{actions on the underlying space are equal}.

\begin{definition}
  An $\alpha$-\emph{graphing} is an equivalence class of $\alpha$-graphing 
  representatives w.r.t. the equivalence relation $\simref$.
\end{definition}

We can now define the notion of abstract program. These are defined as
graphings 

\begin{definition}
  Given an \amc $\alpha:M\acton\Space{X}$, an \emph{$\alpha$-program} 
  $A$ is a $\bar{\alpha}$-graphing $G^{A}$ w.r.t. the monoid action 
  $\bar{\alpha}=\alpha\times\mathfrak{S}_{k}\acton \Space{X}\times\Space{S^{A}}$,
  where $\Space{S^{A}}$ is a finite set of \emph{control states} of cardinality
  $k$ and $\mathfrak{S}_{k}$ is the group of permutations of $k$ elements.
\end{definition}

Now, as a sanity check, we will show how the notion of graphing do capture the
expected dynamics. For this, we restrict to \emph{deterministic graphings}, and
show the notion relates to the usual notion of dynamical system.

\begin{definition}
  An $\alpha$-graphing representative $G$ is deterministic if for all $x\in\Space{X}$ 
  there is at most one $e\in\ E^{G}$ such that $x\in S^{G}_{e}$.
  An $\alpha$-graphing is \emph{deterministic} if its representatives are deterministic.
   An abstract program is \emph{deterministic} if its underlying graphing is
  deterministic.
 \end{definition}
 
 \begin{lemma}
 There is a one-to-one correspondence between the set of deterministic graphings
 w.r.t. the action $M\acton\Space{X}$ and the set of partial 
 dynamical systems $f:\Space{X}\hookrightarrow\Space{X}$ whose graph
 is contained in the preorder\footnote{When $\alpha$ is a group action
 acting by measure-preserving transformations, this is a \emph{Borel 
 equivalence relation} $\mathcal{R}$, and the condition stated here boils
 down to requiring that $f$ belongs to the \emph{full group} of $\alpha$.}
 $\{(x,y)\mid \exists m\in M, \alpha(m)(x)=y\}$.
 \end{lemma}
 
 Lastly, we define some restrictions of $\alpha$-programs that will be important
 later. First, we will restrict the possible subspaces considered as sources of the
 edges, as unrestricted $\alpha$-programs could compute even undecidable 
 problems by, e.g. encoding it into a subspace used as the source of an edge.
 Given an integer $k\in\omega$, we define the following subspaces of 
$\realN^{\omega}$, for $\star\in\{>,\geqslant,=,\neq,\leqslant,<\}$:
\[ \realN^{\omega}_{k\star 0}=\{(x_{1},\dots,x_{k},\dots)\in
	\realN^{\omega}\mid x_{k}\star 0\}.\]

\begin{definition}[Computational graphings] 
Let $\AMC{G}{R}{\alpha}$ be an \amc. 
A \emph{computational $\alpha$-graphing} is an $\alpha$-graphing $T$ with
distinguished states $\top, \bot \in \Space{S^{A}}$ which admits a finite
representative such that each edge $e$ has its source equal to one among
$\realN^{\omega}$, 
$\realN^{\omega}_{k\geqslant 0}$, 
$\realN^{\omega}_{k\leqslant 0}$, 
$\realN^{\omega}_{k> 0}$, 
$\realN^{\omega}_{k< 0}$, 
$\realN^{\omega}_{k=0}$, and 
$\realN^{\omega}_{k\neq 0}$.
\end{definition}

\begin{definition}[treeings] 
Let $\AMC{G}{R}{\alpha}$ be an \amc. 
An \emph{$\alpha$-treeing} is an acyclic 
and finite $\alpha$-graphing, i.e. an $\alpha$-graphing $F$ for which 
there exists a finite $\alpha$-graphing representative $T$ whose set of control states 
$\Space{S^{T}}=\{0,\dots,s\}$ can be endowed with an order $<$ such that every edge of 
$T$ is state-increasing, i.e. for each edge $e$ of source $S_{e}$, for all $x\in S_{e}$, 
\[ \pi_{\Space{S^{T}}}(\alpha(m_{e})(x))>\pi_{\Space{S^{T}}}(x),\] where 
$\pi_{\Space{S^{T}}}$ denotes the projection onto the control states space.

A \emph{computational $\alpha$-treeing} is an $\alpha$-treeing $T$ which is a 
computational $\alpha$-graphing with the distinguished states $\top$, $\bot$ being 
incomparable maximal elements of the state space.
\end{definition}

\subsection{Quantitative Soundness}

As mentioned in the introduction, we will use the property of 
\emph{quantitative soundness} of the dynamic semantics just introduced.
This result is essential, as it connects the time complexity of programs in
the model considered (e.g. \prams, algebraic computation trees) with the
length of the orbits of the considered dynamical system. We here only state 
quantitative soundness for \emph{computational graphings}, i.e. graphings
that have distinguished states $\top$ and $\bot$ representing acceptance
and rejection respectively. In other words, we consider graphings which 
compute \emph{decision problems}.

Quantitative soundness is expressed with respect to a translation of machines
as graphings, together with a translation of inputs as points of the configuration 
space. In the following section, these operations are defined for each model
of computation considered in this paper. In all these cases, the representation 
of inputs is straightforward.

\begin{definition}
Let $\alpha$ be an abstract model of computation, and $\mathbb{M}$ a
model of computation. A \emph{translation} of $\mathbb{M}$ w.r.t. $\alpha$ is
a pair of maps $\Interpret{\cdot}$ which associate to each machine $M$ in 
$\mathbb{M}$ computing a decision problem a computational $\alpha$-graphing 
$\Interpret{M}$ and to each input $\iota$ a point $\Interpret{\iota}$ in 
$\Space{X}\times\Space{S}$.
\end{definition}

\begin{definition}
Let $\alpha$ be an abstract model of computation, $\mathbb{M}$ a model 
of computation. The \amc $\alpha$ is \emph{quantitatively sound} for $\mathbb{M}$ 
w.r.t. a translation $\Interpret{\cdot}$ if for all 
machine $M$ computing a decision problem and input $\iota$, $M$ accepts $\iota$
(resp. rejects $\iota$) in $k$ steps if and only if $\Interpret{M}^{k}(\Interpret{\iota})=\top$ 
(resp. $\Interpret{M}^{k}(\Interpret{\iota})=\bot$).
\end{definition}

\subsection{The algebraic \amcs}

We now define the actions $\amcfull$ and $\amcrealfull$. Those will capture all algebraic
models of computation considered in this paper, and the main theorem will be
stated for this monoid action.

As we intend to consider \prams at some point, we consider from the beginning 
the memory of our machines to be separated in two infinite blocks
$\integerN^{\omega}$, intended to represent both \emph{shared} and a
\emph{private} memory cells\footnote{Obviously, this could be done without any
explicit separation of the underlying space, but this will ease the constructions 
of the next section.}. 

\begin{definition}
The underlying space of $\amcfull$ is
$\Space{X}= \integerN^{\integerN}\cong\integerN^{\omega}\times
\integerN^{\omega}$. The set of generators
is defined by their action on the underlying space, writing 
$k//n$ the floor $\floor{k/n}$ of $k/n$ with the convention that $k//n=0$ 
when $n=0$:
\begin{itemize}
\item $\boconst{i}{c}$ initialises the register $i$ with the constant $c\in\integerN$: $\amcfull(\boconst{i}{c})(\vec{x}) =
  (\vec{x}\{x_i:= c\})$;
\item $\bogen{i}{j}{k}$ ($\star\in\{+,-,\times,//\}$) performs the algebraic operation $\star$ on the values in registers $j$ and $k$ and store the result in register $i$:  $\amcfull(\bogen{i}{j}{k})(\vec{x}) =
  (\vec{x}\{x_i:= x_{j}\star x_k\})$;
\item $\bogenconst{i}{j}{c}$ ($\star\in\{+,-,\times,//\}$) performs the algebraic operation $\star$ on the value in register $j$ and the constant $c\in\integerN$ and store the result in register $i$:  $\amcfull(\bogenconst{i}{j}{c})(\vec{x}) =
  (\vec{x}\{x_i:= c\star x_{j}\})$;
\item $\copyy{i}{j}$ copies the value stored in register $j$ in register $i$: $\amcfull(\copyy{i}{j})(\vec{x}) =
  (\vec{x}\{x_i:= x_j\})$;
\item $\refcopy{i}{j}$ copies the value stored in register $j$ in the register whose index is the value stored in register $i$: $\amcfull(\refcopy{i}{j})(\vec{x}) =
  (\vec{x}\{x_{x_i}:= x_j\})$;
\item $\copyref{i}{j}$ copies the value stored in the register whose index is the value stored in register $j$ in register $i$: $\amcfull(\copyref{i}{j})(\vec{x}) =
  (\vec{x}\{x_i:= x_{x_j}\})$;
\item $\bosqrtn{i}{j}$ computes the floor of the $n$-th root of the value stored in register $j$ and store the result in register $i$: $\amcfull(\bosqrtn{i}{j})(\vec{x}) =
  (\vec{x}\{x_i:= \sqrt[n]{x_j}\})$.
\end{itemize}
\end{definition}

We also define the real-valued equivalent, which will be essential for the
proof of lower bounds. The corresponding \amc $\amcrealfull$ is defined in 
the same way than the integer-valued one, but with underlying space
\(\Space{X}= \realN^{\integerN}\) and with instructions adapted accordingly:
\begin{itemize}[noitemsep,nolistsep]
\item the division and $n$-th root operations are the usual operations on the reals;
\item the three copy operators are only effective on integers.
\end{itemize}

\begin{definition}
The underlying space of $\amcrealfull$ is
$\Space{X}= \realN^{\integerN}\cong\realN^{\omega}\times
\realN^{\omega}$. The set of generators
is defined by their action on the underlying space, with the convention that $k/n=0$ 
when $n=0$:
\begin{itemize}
\item $\boconst{i}{c}$ initialises the register $i$ with the constant $c\in\realN$: $\amcrealfull(\boconst{i}{c})(\vec{x}) =
  (\vec{x}\{x_i:= c\})$;
\item $\bogen{i}{j}{k}$ ($\star\in\{+,-,\times,/\}$) performs the algebraic operation $\star$ on the values in registers $j$ and $k$ and store the result in register $i$:  $\amcrealfull(\bogen{i}{j}{k})(\vec{x}) =
  (\vec{x}\{x_i:= x_{j}\star x_k\})$;
\item $\bogenconst{i}{j}{c}$ ($\star\in\{+,-,\times,/\}$) performs the algebraic operation $\star$ on the value in register $j$ and the constant $c\in\realN$ and store the result in register $i$:  $\amcrealfull(\bogenconst{i}{j}{c})(\vec{x}) =
  (\vec{x}\{x_i:= c\star x_{j}\})$;
\item $\copyy{i}{j}$ copies the value stored in register $j$ in register $i$: $\amcrealfull(\copyy{i}{j})(\vec{x}) =
  (\vec{x}\{x_i:= x_j\})$;
\item $\refcopy{i}{j}$ copies the value stored in register $j$ in the register whose index is the floor of the value stored in register $i$: $\amcrealfull(\refcopy{i}{j})(\vec{x}) =
  (\vec{x}\{x_{\floor{x_i}}:= x_j\})$;
\item $\copyref{i}{j}$ copies the value stored in the register whose index is the floor of the value stored in register $j$ in register $i$: $\amcrealfull(\copyref{i}{j})(\vec{x}) =
  (\vec{x}\{x_i:= x_{\floor{x_j}}\})$;
\item $\bosqrtn{i}{j}$ computes the $n$-th real root of the value stored in register $j$ and store the result in register $i$: $\amcrealfull(\bosqrtn{i}{j})(\vec{x}) =
  (\vec{x}\{x_i:= \sqrt[n]{x_j}\})$.
\end{itemize}
\end{definition}

\section{Algebraic models of computations as \amcs}

\subsection{Algebraic computation trees}


The first model considered here will be that of \emph{algebraic computation
tree} as defined by Ben-Or \cite{Ben-Or83}. Let us note this model refines
the \emph{algebraic decision trees} model of Steele and Yao \cite{SteeleYao82},
a model of computation consisting in binary trees for which each branching
performs a test w.r.t. a polynomial and each leaf is labelled $\texttt{YES}$ or 
$\texttt{NO}$. Algebraic computation trees only allow tests w.r.t. $0$, while 
additional vertices corresponding to algebraic operations can be used to
construct polynomials.

\begin{definition}[algebraic computation trees, {\cite{Ben-Or83}}]
  An \emph{algebraic computation tree} on $\realN^n$ is a binary tree $T$ with a
  function assigning:
  \begin{itemize}
  \item to any vertex $v$ with only one child (simple vertex) an operational
    instruction of the form
    $f_v = f_{v_i} \star f_{v_j}$,
     $f_v = c \star f_{v_i}$, or
     $f_v = \sqrt{f_{v_i}}$,
    where $\star \in \{ +, -, \times, /\}$, ${v_i}, {v_j}$ are ancestors of $v$
    and $c\in \realN$ is a constant;
  \item to any vertex $v$ with two children a test instruction of the form
   $
      f_{v_i} \star 0
   $,
    where $\star \in \{ >,=,\geqslant \}$, and $v_i$ is an ancestor of $v$ 
    or $f_{v_i} \in \{x_1,\dots,x_n\}$;
  \item to any leaf an output $\texttt{YES}$ or $\texttt{NO}$.
  \end{itemize}
\end{definition}


Let $W \subseteq \realN^n$ be any set and $T$ be an algebraic computation
tree. We say that $T$ computes the membership problem for $W$ if for all $x \in
\realN^n$, the traversal of $T$ following $x$ ends on a leaf labelled
\texttt{YES} if and only if $x\in W$.

As algebraic computation trees are \emph{trees}, they will be represented by treeings,
i.e. $\amcrealfull$-programs whose set of control states can be ordered so that any
edge in the graphing is strictly increasing on its control states component.

\begin{definition}
Let $T$ be a computational $\amcrealfull$-treeing. The set of inputs $\inputs{T}$ 
(resp. outputs $\outputs{T}$) is the set of integers $k$ (resp. $i$) such that there 
exists an edge $e$ in $T$ satisfying that:
\begin{itemize} 
\item either $e$ is realised by one of $\boadd{i}{j}{k}$, $\boadd{i}{k}{j}$,
$\bosubstract{i}{j}{k}$, $\bosubstract{i}{k}{j}$, $\bomultiply{i}{j}{k}$, 
$\bomultiply{i}{k}{j}$, $\bodivide{i}{j}{k}$, $\bodivide{i}{k}{j}$
$\boaddconst{i}{k}{c}$, $\bosubstractconst{i}{k}{c}$, $\bomultiplyconst{i}{k}{c}$,
$\bodivideconst{i}{k}{c}$, $\bosqrtn{i}{k}$;
\item or the source of $e$ is one among 
$\realN^{\omega}_{k\geqslant 0}$, 
$\realN^{\omega}_{k\leqslant 0}$, 
$\realN^{\omega}_{k> 0}$, 
$\realN^{\omega}_{k< 0}$, 
$\realN^{\omega}_{k=0}$, and 
$\realN^{\omega}_{k\neq 0}$.
\end{itemize}

The \emph{effective input space} $\inputspaceE{T}$ of an $\amcact$-treeing $T$ is defined 
as the set of indices $k\in\omega$ belonging to $\inputs{T}$ but not to $\outputs{T}$. The
\emph{implicit input space} $\inputspaceI{T}$ of an $\amcact$-treeing $T$ is defined 
as the set of indices $k\in\omega$ such that $k\not \in \outputs{T}$.
\end{definition}

\begin{definition}
Let $T$ be an $\amcrealfull$-treeing, and assume that $1,2,\dots,n \in \inputspaceI{T}$. We say 
that $T$ computes the membership problem for $W\subseteq\realN^{n}$ in $k$ steps if $k$ 
successive iterations of $T$ restricted to 
$\{(x_i)_{i\in\omega} \in \realN^\omega \mid \forall 1\leqslant i\leqslant n, x_i=y_i\}\times\{0\}$ 
reach state $\top$ if and only if $(y_1,y_2,\dots,y_n)\in W$. 
\end{definition}

\begin{remark}
Let $\vec{x}=(x_1,x_2,\dots,x_n)$ be an element of $\realN^n$ and consider two elements 
$a,b$ in the subspace $\{(y_1,\dots, y_n,\dots)\in\realN^{\omega}
\mid\forall 1\geqslant i\geqslant n, y_i=x_i\}\times\{0\}$. One easily checks
that $\pi_{\Space{S}}(T^{k}(a))=\top$ if and only if $\pi_{\Space{S}}(T^{k}(b))
=\top$, where $\pi_{\Space{S}}$ is the projection onto the state space and 
$T^{k}(a)$ represents the $k$-th iteration of $T$ on $a$. It is therefore
possible to consider only a standard representative $\Interpret{\vec{x}}$ of $\vec{x} \in 
\realN^n$, for instance $(x_1,\dots, x_n,0,0,\dots) \in \realN^\omega$, to
decide whether $\vec{x}$ is accepted by $T$.
\end{remark}

\begin{definition}
  Let $T$ be an algebraic computation tree on $\realN^{n}$, and $T^{\circ}$ 
  be the associated directed acyclic graph, built from $T$ by merging all the 
  leaves tagged $\texttt{YES}$ in one leaf $\top$ and all the leaves tagged 
  $\texttt{NO}$ in one leaf $\bot$. Suppose the internal vertices are numbered
  $\{n+1,\dots, n+\ell \}$; the numbers $1,\dots,n$ being reserved for the input.

  We define $\Interpret{T}$ as the $\amcact$-graphing with control states
  $\{n+1,\dots, n+\ell, \top, \bot \}$ and where each internal vertex $i$ of
  $T^{\circ}$ defines either:
  \begin{itemize}
  \item a single edge of source $\realN^{\omega}$ realized by:
    \begin{itemize}
    \item $(\bogen{i}{j}{k}, i \mapsto t)$ ($\star\in\{+,-,\times\}$) if $i$ is associated to $f_{v_i} =
      f_{v_j} \star f_{v_k}$ and $t$ is the child of $i$;
    \item $(\bogenconst{i}{j}{c}, i \mapsto t)$ ($\star\in\{+,-,\times\}$) if $i$ is associated to $f_{v_i} =
     c \star f_{v_k}$ and $t$ is the child of $i$;
    \end{itemize}
  \item a single edge of source $\realN^{\omega}_{k\neq 0}$ realized by:
    \begin{itemize}
    \item $(\bodivide{i}{j}{k}, i \mapsto t)$ if $i$ is associated to
      $f_{v_i} = f_{v_j} / f_{v_k}$ and $t$ is the child of $i$;
    \item $(\bodivideconst{i}{k}{c}, i \mapsto t)$ if $i$ is associated to
      $f_{v_i} = c / f_{v_k}$ and $t$ is the child of $i$;
    \end{itemize}
  \item a single edge of source $\realN^{\omega}_{k\geqslant 0}\times\{i\}$ 
  realized by $(\bosqrt[2]{i}{k}, i \mapsto t)$ if $i$ is associated to 
  $f_{v_i} = \sqrt{f_{v_k}}$ and $t$ is the child of $i$;
  \item two edges if $i$ is associated to 
      $f_{v_i} \star 0$ (where $\star$ ranges in $>$, $\geqslant$) and its two sons 
      are $j$ and $k$. Those are of respective sources
      $\realN^{\omega}_{k\star 0}\times \{i\}$
      and 
      $\realN^{\omega}_{k\bar{\star} 0}
      \times \{i\}$ (where $\bar{\star}='\leqslant'$ if $\star='>'$, $\bar{\star}='<'$ if 
      $\star='\geqslant'$, and $\bar{\star}='\neq'$ if $\star='='$.),
      respectively realized by $(\identity, i \mapsto j)$ and $(\identity, i \mapsto k)$ 
  \end{itemize}
\end{definition}

\begin{proposition}\label{prop:fullyact}
Any algebraic computation tree $T$ of depth $k$ is faithfully and quantitatively 
interpreted as the $\amcrealfull$-program $\Interpret{T}$.
I.e. $T$ computes the membership problem for $W\subseteq\realN^{n}$ if
and only if $\Interpret{T}$ computes the membership problem for $W$ in $k$ steps
-- that is $\pi_{\Space{S}}(\Interpret{T}^k(\Interpret{\vec{x}}))=\top$.
%
\end{proposition}

As a corollary of this proposition, we get quantitative soundness.

\begin{theorem}\label{thm:soundACT}\label{thm:quantsoundact}
  The representation of \acts as $\amcrealfull$-programs is quantitatively sound.
\end{theorem}

\subsection{Algebraic circuits}

As we will recover Cucker's proof that $\NCReal\neq\PtimeReal$, we
introduce the model of \emph{algebraic circuits} and their representation 
as $\amcrealfull$-programs.

\begin{definition}
An algebraic circuit over the reals with inputs in $\realN^n$ is a finite 
directed graph whose vertices have labels in $\naturalN\times\naturalN$, 
that satisfies the following conditions:
\begin{itemize}
\item There are exactly $n$ vertices $v_{0,1},v_{0,2},\dots,v_{0,n}$ with 
first index $0$, and they have no incoming edges;
\item all the other vertices $v_{i,j}$ are of one of the following types:
\begin{enumerate}
\item arithmetic vertex: they have an associated arithmetic operation 
$\{+,-,\times,/\}$ and there exist natural numbers $l, k, r, m$ with $l, k < i$ 
such that their two incoming edges are of sources $v_{l,r}$ and $v_{k,m}$;
\item constant vertex: they have an associated real number $y$ and no 
incoming edges;
\item sign vertex: they have a unique incoming edge of source $v_{k,m}$ 
with $k < i$.
\end{enumerate}
\end{itemize}
We call \emph{depth} of the circuit the largest $m$ such that there exist 
a vertex $v_{m,r}$, and \emph{size} of the circuit the total number of vertices. 
A circuit of depth $d$ is \emph{decisional} if there is only one vertex $v_{d,r}$ 
at level $d$, and it is a sign vertex; we call $v_{d,r}$ the \emph{end vertex} of 
the decisional circuit. 
\end{definition}

To each vertex $v$ one inductively associates a function $f_v$ of the input 
variables in the usual way, where a sign node with input $x$ returns $1$ if 
$x > 0$ and $0$ otherwise. The accepted set of a decisional circuit 
$C$ is defined as the set $S\subseteq\realN^n$ of points whose image by the 
associated function is $1$, i.e. $S=f_v^{-1}(\{1\})$ where $v$ is the end vertex 
of $C$. 
 
We represent algebraic circuit as computational $\amcrealfull$-treeings as follows. 
The first index in the pairs $(i,j)\in\naturalN\times\naturalN$ are represented as 
states, the second index is represented as an index in the infinite product 
$\realN^\omega$, and vertices are represented as edges.

\begin{definition}
Let $C$ be an algebraic circuit, defined as a finite directed graph $(V,E,s,t,\ell)$ 
where $V\subset\naturalN\times\naturalN$, and 
$\ell:V\rightarrow \{\mathrm{init},+,-,\times,/,\mathrm{sgn}\}\cup\{\mathrm{const}_c\mid c\in \realN\}$ 
is a vertex labelling map. We suppose without loss of generality that for each 
$j\in\naturalN$, there is at most one $i\in\naturalN$ such that $(i,j)\in V$. We 
define $N$ as $\max\{j\in\naturalN \mid \exists i\in \naturalN, (i,j)\in V\}$.

We define the $\amcrealfull$-program $\Interpret{C}$ by choosing as set of control 
states $\{ i\in \naturalN \mid \exists j\in \naturalN, (i,j)\in V\}$ and the collection 
of edges $\{e_{(i,j)}\mid i\in\naturalN^*, j\in\naturalN, (i,j)\in V\}\cup\{e^+_{(i,j)}\mid i\in\naturalN^*, j\in\naturalN, (i,j)\in V, \ell(v)=\mathrm{sgn}\}$ 
realised as follows:
\begin{itemize}
\item if $\ell(v)=\mathrm{const}_c$, the edge $e_{(i,j)}$ is realised as 
$(\boaddconst{j}{c}{n_v}, 0 \mapsto i)$ of source $\realN^{\omega}_{n_v= 0}\times \{0\}$;
\item if $\ell(v)=\star$ ($\star\in\{+,-,\times\}$) of incoming edges $(k,l)$ and $(k',l')$, 
the edge $e_{(i,j)}$ is of source $\realN^\omega\times\{\max(k,k')\}$ and realised by 
$(\bogen{j}{l}{l'}, \max(k,k') \mapsto i)$;
\item if $\ell(v)=/$ of incoming edges $(k,l)$ and $(k',l')$, the edge $e_{(i,j)}$ is 
of source $\realN^\omega_{l'\neq 0}\times\{\max(k,k')\}$ and realised by 
$(\bodivide{j}{l}{l'}, \max(k,k') \mapsto i)$;
\item if $\ell(v)=\mathrm{sgn}$ of incoming edge $(k,l)$, the edges $e_{(i,j)}$ 
and $e^+_{(i,j)}$ are of respective sources
 $\realN^\omega_{n_v=0\wedge x_l\leqslant 0}\times\{k\}$ 
and $\realN^\omega_{n_v=0\wedge x_l>0}\times\{k\}$ realised by 
$(\identity, k \mapsto i)$ and  $(\boadd{j}{n_v}{1}, k \mapsto i)$ respectively.
\end{itemize}
\end{definition}

As each step of computation in the algebraic circuit is translated as going
through a single edge in the corresponding $\amcrealfull$-program, the following
result is straightforward.

\begin{theorem}\label{thm:soundalgcirc}\label{thm:quantsoundalgcirc}
  The representation of \algcirc as $\amcrealfull$-programs is quantitatively sound.
\end{theorem}

\subsection{Algebraic \srams}
\label{sec:algebraic-prams}

In this paper, we will consider algebraic parallel random access machines, that
act not on strings of bits, but on integers. In order to define those properly,
we first define the notion of (sequential) random access machine (\sram) before
considering their parallelisation.

A \sram \emph{command} is a pair $(\ell, I)$ of a \emph{line}
$\ell \in \naturalN^{\star}$ and an \emph{instruction} $\command{I}$ among the
following, where $i, j \in \naturalN$, $\star \in \{+, -, \times, / \}$, $c\in \integerN$
is a constant and $\ell, \ell' \in \naturalN^{\star}$ are lines:
\[\begin{array}{c}
\Skip;\hspace{0.8em}  \command{X_{i} \coloneqq c};\hspace{0.8em}
   \command{X_{i} \coloneqq X_{j} \star X_{k}};\hspace{0.8em}
                     \command{X_{i} \coloneqq X_{j}};\\
                         \command{X_{i} \coloneqq \sharp{}X_{j}};\hspace{0.8em}
                                  \command{\sharp{X_{i}}\coloneqq X_{j}};\hspace{0.8em}
                    \conditional{X_i = 0}{\ell}{\ell'}.
                    \end{array}
\]

A \sram \emph{machine} $M$ is then a finite set of commands such that the set of
lines is $\{1,2,\dots,\length{M}\}$, with $\length{M}$ the \emph{length} of
$M$. We will denote the commands in $M$ by $(i,\instr{i})$, i.e. $\instr{i}$
denotes the line $i$ instruction. 

Following Mulmuley \cite{Mulmuley99}, we will here make the assumption that the input
in the \sram (and in the \pram model defined in the next section) is split into \emph{numeric} 
and \emph{nonumeric} data -- e.g. in the \maxflow problem the 
nonnumeric data would specify the network and the numeric data would specify the 
edge-capacities -- and that indirect references use pointers depending only on 
nonnumeric data\footnote{Quoting Mulmuley: "We assume that the pointer involved in an 
indirect reference is not some numeric argument in the input or a quantity that depends on
it. For example, in the max- flow problem the algorithm should not use an edge-capacity as 
a pointer—which is a reasonable condition. To enforce this restriction, one initially puts an 
invalid-pointer tag on every numeric argument in the input. During the execution of an 
arithmetic instruction, the same tag is also propagated to the result if any operand has that 
tag. Trying to use a memory value with invalid-pointer tag results in error." 
\cite[Page 1468]{Mulmuley99}.}.
We refer the reader to Mulmuley's article for more details.

Machines in the \sram model can be represented as graphings w.r.t. the action $\amcfull$.
Intuitively the encoding works as follows. The notion of \emph{control state} allows to
represent the notion of \emph{line} in the program. Then, the action just
defined allows for the representation of all commands but the conditionals. The
conditionals are represented as follows: depending on the value of $X_{i}$ one
wants to jumps either to the line $\ell$ or to the line $\ell'$; this is easily
modelled by two different edges of respective sources
$\hyperplan{i}=\{\vec{x}~|~ x_{i}=0\}$ and
$\complementset{\hyperplan{i}}=\{\vec{x}~|~ x_{i}\neq0\}$.

\begin{definition}
  Let $M$ be a \sram machine. We define the translation $\Interpret{M}$ as the $\amcram$-program with set of
  control states $\{0,1,\dots,L,L+1\}$ where each line $\ell$ defines (in the following, 
  $\star\in\{+,-,\times\}$ and we write $\ell\plusplus$ the map $\ell\mapsto\ell+1$):
  \begin{itemize}[nolistsep,noitemsep]
  \item a single edge $e$ of source $\Space{X}\times\{\ell\}$ and realised
    by:
    \begin{itemize}[noitemsep,nolistsep]
    \item $(\identity,\ell\plusplus)\text{ if }\instr{\ell}=\Skip$;
    \item $(\boconst{i}{c},\ell\plusplus)\text{ if }\instr{\ell}=\command{X_i \coloneqq c}$;
    \item $(\bogen{i}{j}{k},\ell\plusplus)\text{ if }\instr{\ell}=
    \command{X_i \coloneqq X_j \star X_k} $;
    \item $(\copyy{i}{j},\ell\plusplus)\text{ if }\instr{\ell}=
    \command{X_{i} \coloneqq X_{j}}$;
    \item $(\copyref{i}{j},\ell\plusplus)\text{ if }\instr{\ell}=
    \command{X_{i}\coloneqq \sharp X_{j}}$;
    \item $(\refcopy{i}{j},\ell\plusplus)\text{ if }\instr{\ell}=
     \command{\sharp{}X_{i} \coloneqq X_{j}}$.
     \end{itemize}
  \item an edge $e$ of source $\hyperplan{k}\times\{\ell\}$ realised
    by $(\boeuclidivide{i}{j}{k},\ell\plusplus)$ if $\instr{\ell}$ is
    $\command{X_i \coloneqq X_j/X_k}$;
  \item a pair of edges $e,\complementset{e}$ of respective sources
    $\hyperplan{i}\times\{\ell\}$ and $\complementset{\hyperplan{i}}\times\{\ell\}$
    and realised by respectively $(\identity,\ell\mapsto \ell^{0})$ and
    $(\identity,\ell\mapsto \ell^{1})$, if the line is a conditional
    $\command{if~X_{i}=0~goto~\ell^{0}~else~\ell^{1}}$.
  \end{itemize}
  The translation $\Interpret{\iota}$ of an input $\iota\in\mathbf{Z}^d$ is
  the point $(\bar{\iota},0)$ where $\bar{\iota}$ is the sequence
  $(\iota_1,\iota_2,\dots,\iota_k,0,0,\dots)$.
\end{definition}

Now, the main result for the representation of \srams is the following. The proof
is straightforward, as each instruction corresponds to exactly one edge, except for
the conditional case (but given a configuration, it lies in the source of at most one 
of the two edges translating the conditional). 

\begin{theorem}\label{thm:soundsram}
  The representation of \srams as $\amcfull$-program is quantitatively sound w.r.t.
  the translation just defined.
\end{theorem}

\subsection{The Crew operation and \prams}
\label{sec:crew}

Based on the notion of \sram, we are now able to consider their parallelisation,
namely \prams. A \pram $M$ is given as a finite sequence of \sram machines
$M_{1},\dots,M_{p}$, where $p$ is the number of \emph{processors} of $M$. Each
processor $M_{i}$ has access to its own, private, set of registers
$(\command{X}_{k}^{i})_{k\geqslant 0}$ and a \emph{shared memory} represented as
a set of registers $(\command{X}_{k}^{0})_{k\geqslant 0}$.

One has to deal with conflicts when several processors try to access the shared
memory simultaneously. We here chose to work with the \emph{Concurrent Read,
  Exclusive Write} (\crew) discipline: at a given step at which several
processors try to write in the shared memory, only the processor with the
smallest index will be allowed to do so. In order to model such parallel
computations, we abstract the \crew at the level of monoids. For this, we
suppose that we have two monoid actions
$\MonGaR{G}{R}\acton \Space{X}\times\Space{Y}$ and
$\MonGaR{H}{Q}\acton \Space{X}\times\Space{Z}$, where $\Space{X}$ represents the
shared memory. We then consider the subset $\#\subset G\times H$ of pairs of
generators that potentially conflict with one another -- the conflict relation.

\begin{definition}[Conflicted sum]
  Let $\MonGaR{G}{R}$, $\MonGaR{G'}{R'}$ be two monoids and
  $\# \subseteq G \times G'$. The \emph{conflicted sum of $\MonGaR{G}{R}$ and
    $\MonGaR{G'}{R'}$ over $\#$}, noted
  $\MonGaR{G}{R} \ast_{\#} \MonGaR{G'}{R'}$, is defined as the monoid with
  generators \((\{1\} \times G) \cup (\{2\} \times G')\) and relations
  \[ \begin{array}{c}(\{1\} \times R) \cup (\{2\} \times R') \cup \{ (\mathbf{1}, e)\} \cup \{ (\mathbf{1},
    e')\} \\
  \hspace{1em}\cup \{ \big((1,g)(2,g'), (2,g')(1,g)
    \big) \mid (g,g') \notin \# \}\end{array}
  \] where $\mathbf{1}$, $e$, $e'$ are the units of
  $\MonGaR{G}{R} \ast_{\#} \MonGaR{G'}{R'}$, $\MonGaR{G}{R}$ and $\MonGaR{G'}{R'}$
  respectively.

  In the particular case where $\# = (G \times H') \cup (H \times G')$, with
  $H, H'$ respectively subsets of $G$ and $G'$, we will write the sum
  $\MonGaR{G}{R}\ncproduct{H}{H'}\MonGaR{G'}{R'}$.
\end{definition}

\begin{remark}
  When the conflict relation $\#$ is empty, this defines the usual direct product
  of monoids. This corresponds to the case in which no conflicts can arise
  w.r.t. the shared memory. In other words, the direct product of monoids
  corresponds to the parallelisation of processes \emph{without shared memory}.

  Dually, when the conflict relation is full ($\# = G \times G'$), this defines the free
  product of the monoids.
\end{remark}

\begin{definition}
  Let $\alpha: M\acton\Space{X}\times \Space{Y}$ be a monoid action. We say that
  an element $m\in M$ is \emph{central relatively to $\alpha$} (or just
  \emph{central}) if the action of $m$ commutes with the first projection
  $\pi_{X}:\Space{X}\times \Space{Y}\rightarrow\Space{X}$,
  i.e.\footnote{Here and in the following, we denote by $;$ the sequential
  composition of functions. I.e. $f;g$ denotes what is usually written $g\circ f$.} 
  $\alpha(m);\pi_{X}=\alpha(m)$; in other words $m$ acts as the identity on
  $\Space{X}$.
\end{definition}

Intuitively, central elements are those that will not affect the shared
memory. As such, only \emph{non-central elements} require care
when putting processes in parallel.

\begin{definition}
  Let $\MonGaR{G}{R}\acton \Space{X}\times\Space{Y}$ be an \amc. We note
  $Z_{\alpha}$ the set of central elements and $\bar{Z}_{\alpha}(G)=\{m\in G~|~ n\not\in Z_{\alpha}\}$.
\end{definition}

\begin{definition}[The \crew of \amcs]
  Let $\alpha: \MonGaR{G}{R}\acton \Space{X}\times\Space{Y}$ and
  $\beta: \MonGaR{H}{Q}\acton \Space{X}\times\Space{Z}$ be \amcs. We define the
  \amc
  $\crew(\alpha,\beta): \MonGaR{G}{R}\ncproduct{\bar{Z}_{\alpha}(G)}{\bar{Z}_{\beta}(G')}\MonGaR{G'}{R'}\acton \Space{X}\times\Space{Y}\times\Space{Z}$
  by letting $\crew(\alpha,\beta)(m,m')=\alpha(m)\ast\beta(m')$ on elements of
  $G\times G'$, where:
  \begin{align}
    \lefteqn{\alpha(m)\ast\beta(m')=}\nonumber\\
    &\left\{
      \begin{array}{ll}
        \Delta;[\alpha(m);\pi_{Y},\beta(m')]  & \text{ if $m\not\in\bar{Z}_{\alpha}(G), m'\in\bar{Z}_{\beta}(G')$,}\\
        \Delta;[\alpha(m),\beta(m');\pi_{Z}] & \text{ otherwise,}
        \end{array}\right.\nonumber
  \end{align}
  with
  $\Delta: 
  \Space{X}\times\Space{Y}\times\Space{Z}\rightarrow\Space{X}\times\Space{Y}\times\Space{X}\times\Space{Z};
  (x,y,z)\mapsto(x,y,x,z)$.
\end{definition}

We can now define \amc of \prams and thus the interpretations of \prams as
abstract programs. For each integer $p$, we define the \amc $\crew^{p}(\amcfull)$. 
This allows the consideration of up to $p$ parallel \srams: the translation of such a
\sram with $p$ processors is defined by extending the translation of \srams
by considering a set of states equal to $L_{1}\times L_{2}\times \dots \times L_{p}$ 
where for all $i$ the set $L_{i}$ is the set of lines of the $i$-th processor.

Now, to deal with arbitrary large \prams, i.e. with arbitrarily large number of
processors, one considers the following \amc defined as a \emph{direct limit}.

\begin{definition}[The \amc of  \prams]
  Let $\alpha: M\acton\Space{X\times X}$ be the \amc $\amcfull$.
  The \amc of \prams is defined as $\amcprams=\varinjlim\crew^{k}(\alpha)$, where
  $\crew^{k-1}(\alpha)$ is identified with a restriction of $\crew^{k}(\alpha)$
  through
  $\crew^{k-1}(\alpha)(m_{1},\dots,m_{k-1})\mapsto
  \crew^{k}(\alpha)(m_{1},\dots,m_{k-1},1)$.
\end{definition}

\begin{remark}\label{remarkcrewp}
 We notice that the underlying space of the \pram{} \amc $\amcprams$ is defined 
as the union \( \cup_{n\in \omega} \integerN^{\omega} \times (\integerN^{\omega})^n\) 
which we will write \( \integerN^{\omega} \times (\integerN^{\omega})^{(\omega)}\).
In practise a given $\amcprams$-program admitting a finite $\amcprams$ representative
will only use elements in $\crew^{p}(\amcfull)$, and can therefore be understood as a 
$\crew^{p}(\alpha)$-program. 
\end{remark}

\begin{theorem}\label{thm:soundpram}\label{thm:quantsoundprams}
  The representation of \prams as $\amcprams$-program is quantitatively sound.
\end{theorem}

This result, here stated for integer-valued \prams, can easily be obtained for
\emph{real-valued} \prams translated as $\amcrealfull$-programs.


%% file: entropy.tex
\section{Entropy and Cells}

\subsection{Topological Entropy}

Topological Entropy is a standard invariant of dynamical system. It is a value representing the average exponential growth rate of the number of orbit segments distinguishable with a finite (but arbitrarily fine) precision. The definition is based on the notion of open covers.

\begin{definition}[Open covers] Given a topological space $\Space{X}$, an \emph{open cover} of $\Space{X}$ is a family $\mathcal{U}=(U_{i})_{i\in I}$ of open subsets of $\Space{X}$ such that $\cup_{i\in I}U_{i}=\Space{X}$. A finite cover $\mathcal{U}$ is a cover whose indexing set is finite. A \emph{subcover} of a cover $\mathcal{U}=(U_{i})_{i\in I}$ is a sub-family $\mathcal{S}=(U_{j})_{j\in J}$ for $J\subseteq I$ such that $\mathcal{S}$ is a cover, i.e. such that $\cup_{j\in J}U_{j}=\Space{X}$.
We will denote by $\opencovers{X}$ (resp. $\finiteopencovers{X}$) the set of all open covers (resp. all finite open covers) of the space $\Space{X}$.
\end{definition}


\begin{definition}
An open cover $\mathcal{U}=(U_{i})_{i\in I}$, together with a continuous function $f:\Space{X}\rightarrow\Space{X}$, defines the inverse image open cover $f^{-1}(\mathcal{U})=(f^{-1}(U_{i}))_{i\in I}$.  Given two open covers $\mathcal{U}=(U_{i})_{i\in I}$ and $\mathcal{V}=(V_{j})_{j\in J}$, we define their join $\mathcal{U}\vee\mathcal{V}$ as the family $(U_{i}\cap V_{j})_{(i,j)\in I\times J}$. 
\end{definition}

\begin{remark}
If $\mathcal{U,V}$ are finite, $f^{-1}(\mathcal{U})$ and $\mathcal{U}\vee\mathcal{V}$ are finite.
\end{remark}

Traditionally \cite{AKmA}, entropy is defined for continuous maps on a compact set. 
However, a generalisation of entropy to non-compact sets can easily be defined by restricting the usual definition to \emph{finite} covers\footnote{This is discussed by Hofer \cite{Hofer75} together with another generalisation based on the Stone-\v{C}ech compactification of the underlying space.}. This is the definition we will use here.

\begin{definition}
Let $\Space{X}$ be a topological space, and $\mathcal{U}=(U_{i})_{i\in I}$ be a finite cover of $\Space {X}$. We define the quantity $H_{\Space{X}}^{0}(\mathcal{U})=\min\{\log_{2}(\card{J})~|~ J\subseteq I, \cup_{j\in J}U_{j}=\Space{X}\}$.
\end{definition}
In other words, if $k$ is the cardinality of the smallest subcover of $\mathcal{U}$, $H^{0}(\mathcal{U})=\log_{2}(k)$.

\begin{definition}\label{def:Hk}
Let $\Space{X}$ be a topological space and $f: \Space{X}\rightarrow\Space{X}$ be a continuous map. For any finite open cover $\mathcal{U}$ of $\Space{X}$, define $H_{\Space{X}}^{k}(f,\mathcal{U})=\frac{1}{k}H_{\Space{X}}^{0}(\mathcal{U}\vee f^{-1}(\mathcal{U})\vee\dots\vee f^{-(k-1)}(\mathcal{U}))$. 
\end{definition}

One can show that the limit $\lim_{n\rightarrow\infty}H_{\Space{X}}^{n}(f,\mathcal{U})$ exists and is finite; it will be noted $h(f,\mathcal{U})$. The topological entropy of $f$ is then defined as the supremum of these values, when $\mathcal{U}$ ranges over the set of all finite covers $\finiteopencovers{X}$.

\begin{definition}\label{def:entropy}
Let $\Space{X}$ be a topological space and $f: \Space{X}\rightarrow\Space{X}$ be a continuous map. The \emph{topological entropy} of $f$ is defined as $h(f)=\sup_{\mathcal{U}\in\finiteopencovers{X}} h(f,\mathcal{U})$.
\end{definition}

\subsection{Graphings and Entropy}

We now need to define the entropy of \emph{deterministic graphings}. As mentioned briefly already, deterministic graphings on a space $\Space{X}$ are in one-to-one correspondence with partial dynamical systems on $\Space{X}$. 
 Thus, we only need to extend the notion of entropy to partial maps, and we can then define the entropy of a graphing $G$ as the entropy of its corresponding map $[G]$. 

Given a finite cover $\mathcal{U}$, the only issue with partial continuous maps is that $f^{-1}(\mathcal{U})$ is not in general a cover. Indeed, $\{f^{-1}(U)~|~U\in \mathcal{U}\}$ is a family of open sets by continuity of $f$ but the union $\cup_{U\in \mathcal{U}}f^{-1}(U)$ is a strict subspace of $\Space{X}$ (namely, the domain of $f$). It turns out the solution to this problem is quite simple: we notice that $f^{-1}(\mathcal{U})$ is a cover of $f^{-1}(\Space{X})$ and now work with covers of subspaces of $\Space{X}$. Indeed, $\mathcal{U}\vee f^{-1}(\mathcal{U})$ is itself a cover of $f^{-1}(\Space{X})$ and therefore the quantity $H_{\Space{X}}^{2}(f,\mathcal{U})$ can be defined as $(1/2)H_{f^{-1}(\Space{X})}^{0}(\mathcal{U}\vee f^{-1}(\mathcal{U}))$.

We now generalise this definition to arbitrary iterations of $f$ by extending Definitions \ref{def:Hk} and \ref{def:entropy} to partial maps as follows.
\begin{definition}\label{def:partialentropy}
Let $\Space{X}$ be a topological space and $f: \Space{X}\rightarrow\Space{X}$ be a continuous partial map. For any finite open cover $\mathcal{U}$ of $\Space{X}$, we define $H_{\Space{X}}^{k}(f,\mathcal{U})=\frac{1}{k}H_{f^{-k+1}(\Space{X})}^{0}(\mathcal{U}\vee f^{-1}(\mathcal{U})\vee\dots\vee f^{-(k-1)}(\mathcal{U}))$.
The \emph{entropy} of $f$ is then defined as $h(f)=\sup_{\mathcal{U}\in\finiteopencovers{X}} h(f,\mathcal{U})$, where $h(f,\mathcal{U})$ is again defined as the limit $\lim_{n\rightarrow\infty}H_{\Space{X}}^{n}(f,\mathcal{U})$.
\end{definition}

We now consider the special case of a graphing $G$ with set of control states $S^{G}$. For an intuitive understanding, one can think of $G$ as the representation of a \pram machine. We focus on the specific open cover indexed by the set of control states, i.e. $\mathcal{S}=(\Space{X}\times\{s\}_{s\in S^{G}})$, and call it \emph{the states cover}. We will now show how the partial entropy $H^{k}(G,\mathcal{S})$ is related to the set of \emph{admissible sequence of states}. Let us define those first.

\begin{definition}
Let $G$ be a graphing, with set of control states $S^{G}$. An admissible sequence of states is a sequence $\mathbf{s}=s_{1}s_{2}\dots s_{n}$ of elements of $S^{G}$ such that for all $i\in\{1,2,\dots,n-1\}$ there exists a subset $C$ of $\Space{X}$ -- i.e. a set of configurations -- such that $G$ contains an edge from $C\times\{s_{i}\}$ to a subspace of $\Space{X}\times\{s_{i+1}\}$.
\end{definition}

\begin{example}
As an example, let us consider the very simple graphing with four control states $a,b,c,d$ and edges from $\Space{X}\times\{a\}$ to $\Space{X}\times\{b\}$, from $\Space{X}\times\{b\}$ to $\Space{X}\times\{c\}$, from $\Space{X}\times\{c\}$ to $\Space{X}\times\{b\}$ and from $\Space{X}\times\{c\}$ to $\Space{X}\times\{d\}$. 
Then the sequences $abcd$ and $abcbcbc$ are admissible, but the sequences $aba$, $abcdd$, and $abcba$ are not.
\end{example}

\begin{lemma}\label{lem:admseqHk}
Let $G$ be a graphing, and $\mathcal{S}$ its states cover. Then for all integer $k$, the set $\admss{k}$ of admissible sequences of states of length $k>1$ is of cardinality $2^{k.H^{k}(G,\mathcal{S})}$.
\end{lemma}

A tractable bound on the number of admissible sequences of states can be obtained by noticing that the sequence $H^{k}(G,\mathcal{S})$ is \emph{sub-additive}, i.e. $H^{k+k'}(G,\mathcal{S})\leqslant H^{k}(G,\mathcal{S})+H^{k'}(G,\mathcal{S})$. A consequence of this is that $H^{k}(G,\mathcal{S})\leqslant kH^{1}(G,\mathcal{S})$. Thus the number of admissible sequences of states of length $k$ is bounded by $2^{k^{2}H^{1}(G,\mathcal{S})}$. We now study how the cardinality of admissible sequences can be related to the entropy of $G$. This is deduced from Lemma \ref{lem:admseqHk} and the following general result (which does not depend on the chosen cover).

\begin{lemma}\label{lemma:coverentropybound}
For all $\epsilon>0$ and all cover $\mathcal{U}$, there exists a natural number $N$ such that $\forall k\geqslant N$, $H^{k}(G,\mathcal{U})<h([G])+\epsilon$.
\end{lemma}

The two previous lemmas combine to give the following.

\begin{lemma}\label{lem:admseqHk2}
Let $G$ be a graphing. Then $\card{\admss{k}} = O(2^{k.h([G])})$ as $k\rightarrow\infty$.
\end{lemma}


\subsection{Cells Decomposition}

Now, let $G$ be a deterministic graphing with its state cover $\mathcal{S}$. We fix $k>2$ and consider the partition $(C[\mathbf{s}])_{\mathbf{s}\in\admss{k}}$ of the space $[G]^{-k+1}(\Space{X})$, where the sets $C[\mathbf{s}]=C[(s_{1}s_{2}\dots s_{k-1},s_{k})]$ are defined inductively as follow:
\begin{itemize}[nolistsep,noitemsep]
\item $C[s_1,s_2]$ is the set $\{x\in\Space{X}\mid [G](x,s_1)\in X\times \{s_2\}\}$;
\item $C[(s_{1}s_{2}\dots s_{k-1},s_{k})]$ is the set $\{x\in\Space{X}\mid \forall  i \in\{2,\dots,k\}, [G]^{i-1}(x, s_1)\in\Space{X}\times\{s_i\}\}$.
\end{itemize}
  



This decomposition splits the set of initial configurations into cells satisfying the following property: \emph{for any two initial configurations contained in the same cell $C[\mathbf{s}]$, the $k$-th first iterations of $G$ go through the same admissible sequence of states $\mathbf{s}$}.


\begin{definition}
Let $G$ be a deterministic graphing, with its state cover $\mathcal{S}$. Given an integer $k$, we define the $k$-th cell decomposition of $\Space{X}$ along $G$ as the partition $\{ C[\mathbf{s}] ~|~ \mathbf{s}\in\admss{k} \}$.
\end{definition}

Then Lemma \ref{lem:admseqHk} provides a bound on the cardinality of the $k$-th cell decomposition.
Using the results in the previous section, we can then obtain the following proposition.

\begin{proposition}\label{prop:entropy-k-cell}
Let $G$ be a deterministic graphing, with entropy $h(G)$. The cardinality of the $k$-th cell decomposition of $\Space{X}$ w.r.t. $G$, as a function $c(k)$ of $k$, is asymptotically bounded by $g(k)=2^{k.h([G])}$, i.e. $c(k)=O(g(k))$.
\end{proposition}

We also state another bound on the number of cells of the $k$-th cell decomposition, based on the state cover entropy, i.e. the entropy with respect to the state cover rather than the usual entropy which takes the supremum of cover entropies when the cover ranges over all finite covers of the space. This is a simple consequence of Lemma \ref{lem:admseqHk}.

\begin{proposition}\label{prop:entropy-k-cell-h0}
Let $G$ be a deterministic graphing. We consider the \emph{state cover entropy} $h_0([G])=\lim_{n\rightarrow\infty}H_{\Space{X}}^{n}([G],\mathcal{S})$ where $\Space{S}$ is the state cover. The cardinality of the $k$-th cell decomposition of $\Space{X}$ w.r.t. $G$, as a function $c(k)$ of $k$, is asymptotically bounded by $g(k)=2^{k.h_0([G])}$, i.e. $c(k)=O(g(k))$.
\end{proposition}

%% file: benor.tex
\section{Entropic Cotrees and the Main Lemma}

\subsection{Lower Bounds through the Milnor-Thom theorem} 

The results stated in
the last section can be used to prove lower bounds in several models. 
These results rely on two ingredients: the above bounds on the cardinality 
of the $k$-th cell decomposition, and the Milnor-Thom theorem. 

The Milnor-Thom theorem, which was proven independently by Milnor 
\cite{Milnor:1964} and Thom \cite{Thom}, states bounds on the sum of the Betti 
numbers (i.e. the rank of the homology groups) of an algebraic variety. This 
theorem provides bounds on the number of connected components (i.e. the 
$0$-th Betti number $\beta_0(V)$) of a semi-algebraic variety $V$. We here 
use the statement of the Milnor-Thom theorem as given by Ben-Or 
\cite[Theorem 2]{Ben-Or83}.

\begin{theorem}
 \label{thm:milnor-ben-or}
Let $V \subseteq \realN^n$ be a set defined by polynomial 
in\pointmedian{}equations
($n,m,h\in\naturalN$):
  \begin{align}
    &\{q_i(\vec{x}) = 0\mid 0\leqslant i\leqslant m\} \nonumber\\
    &\hspace{1em}\cup \{p_i(\vec{x}) > 0\mid 0\leqslant i\leqslant s\} \nonumber\\
    &\hspace{2em}\cup \{p_i(\vec{x}) \leqslant 0\mid s+1\leqslant i\leqslant h\}.\nonumber
   \end{align}
Then $\beta_0(V)$ is at most $d(2d-1)^{n+h-1}$,
where $d=\max\{2, \deg(q_i),\deg(p_j)\}$.
\end{theorem}


\noindent The lower bounds proofs then proceed by the following proof strategy:
\begin{enumerate}[noitemsep]
\item consider an algebraic model of computation, and define the 
corresponding \amc;
\item show that the cells in the $k$-th cell decomposition are 
semi-algebraic sets defined by systems of equations $E$ with explicit 
\emph{upper bounds} on the number of equations and the degrees of the 
polynomials;
\item bound the number of connected components of each cell by the 
Milnor-Thom theorem;
\item given an algebraic problem (e.g. a subset of $\realN^k$), deduce lower 
bounds on the length of the computations deciding that problem based on its 
number of connected components.
\end{enumerate}

Among the lower bound proofs using this proof strategy, we point out 
Steele and Yao lower bounds on algebraic decision trees \cite{SteeleYao82}, 
and Mulmuley's proof of lower bounds on ``\prams without bit operations'' 
\cite{Mulmuley99}. These results do not use the notion of entropy. 
Due to space constraints, we do not detail these in this paper.

We will now explain how this method can be refined following Ben-Or's proof of 
lower bounds for algebraic computational trees. Indeed, while Mulmuley's 
\cite{Mulmuley99} was not later improved upon, Steele and Yao's lower bounds were 
extended by Ben-Or \cite{Ben-Or83} to encompass algebraic computational trees with sums, 
substractions, products, divisions and square roots. The technique of Ben-Or 
improves on Steele and Yao in that it provides a method to deal with divisions and 
square roots. We here abstract this method by considering \emph{$k$-th entropic 
co-trees} which are a refinement of the $k$-th cell decomposition. This 
allows us to recover Ben-Or's result, capture Cucker's proof that $\NCReal\neq\PtimeReal$, 
and to strengthen Mulmuley's result by allowing the machines considered to use divisions 
and square roots.

\subsection{Entropic co-trees}\label{subsec:cotrees}
  
The principle underlying the improvement of Ben-Or on Steele and Yao
consists in adding additional variables to avoid using the square root or division, 
obtaining in this way a system of polynomial equations instead of a single 
equation for a given cell in the $k$-th cell decomposition. For instance, instead 
of writing the equation $p/q<0$, one defines a fresh variable $r$ and considers 
the system $\{p=qr; r<0\}$.

To adapt it to graphings, we consider the notion of \emph{entropic co-tree} of
a graphing that generalises the $k$-th cell decomposition to account for the 
instructions used at each step of the computation.

As we explained in \Cref{remarkcrewp}, a given \pram is interpreted as a
$\crew^p(\amcrealfull)$-program for a fixed integer $p$ (the number of 
processors). It is therefore enough to state the following definitions and 
results for the \amc $\crew^p(\amcrealfull)$ to apply them to the interpretations
of arbitrary \prams.

\begin{definition}[$k$-th entropic co-tree]
\label{def:co-tree}
Consider a deterministic $\crew^p(\amcrealfull)$-graphing representative $T$, and fix an element $\top$ 
of the set of control states. We can define the $k$-th entropic co-tree of $T$ 
along $\top$ and the state cover inductively:
\begin{itemize}
\item $k=0$, the co-tree $\cotree{0}$ is simply the root 
$n^{\epsilon}=\realN^n\times\{\top\}$;
\item $k=1$, one considers the preimage of $n^{\epsilon}$ through $T$, i.e. 
$T^{-1}(\realN^n\times\{\top\})$ the set of all non-empty sets 
$\alpha(m_{e})^{-1}(\realN^n\times\{\top\})$ and intersects it pairwise with the 
state cover, leading to a finite family (of cardinality bounded by the number of 
states multiplied by the number of edges fo $T$) $(n_{e}^{i})_{i}$ 
defined as $n^{i}=T^{-1}(n^{\epsilon})\cap \realN^n\times\{i\}$. The first entropic
co-tree  $\cotree{1}$ of $T$ is then the tree defined by linking each $n_{e}^{i}$ 
to $n^{\epsilon}$ with an edge labelled by $m_{e}$;
\item $k+1$, suppose defined the $k$-th entropic co-tree of $T$, defined as a 
family of elements $n_{\seq{e}}^{\pi}$ where $\pi$ is a finite sequence of states of 
length at most $k$ and $\seq{e}$ a sequence of edges of $T$ of the same
length, and where $n_{\seq{e}}^{\pi}$ and $n_{\seq{e'}}^{\pi'}$ are linked by an 
edge labelled $f$ if and only if $\pi'=\pi.s$ and $\seq{e'}=f.\seq{e}$ where $s$ is 
a state and $f$ an edge of $T$. We consider the subset of elements 
$n_{\seq{e'}}^{\pi}$ where $\pi$ is exactly of length $k$, and for each such 
element we define new vertices $n_{e.\seq{e'}}^{\pi.s}$ defined as 
$\alpha(m_{e})^{-1}(n_{\seq{e'}}^{\pi})\cap \realN^n\times\{s\}$ when it is non-empty.
The $k+1$-th entropic co-tree $\cotree{k+1}$ is defined by extending the $k$-th 
entropic co-tree  $\cotree{k}$, adding the vertices $n_{e.\seq{e'}}^{\pi.s}$ and linking 
them to $n_{\seq{e'}}^{\pi}$ with an edge labelled by $e$.
\end{itemize}
\end{definition}


We can easily obtain bounds on the size of the cotrees, refining the bounds
on the $k$th cell decomposition.
\begin{proposition}\label{prop:entropy-k-cell-h0-edges}
Let $G$ be a deterministic $\crew^{p}(\amcrealfull)$-graphing with a finite set of edges $E$, and 
$\seqedges{k}{E}$ the set of length $k$ sequences of edges in $G$. 
We consider the \emph{state cover entropy} 
$h_0([G])=\lim_{n\rightarrow\infty}H_{\Space{X}}^{n}([G],\mathcal{S})$ where 
$\Space{S}$ is the state cover. The cardinality of the length $k$ vertices of the 
entropic co-tree of $G$, as a function $c(k)$ of $k$, is asymptotically bounded 
by $g(k)=\card{\seqedges{k}{E}}.2^{k.h_0([G])}$, which is itself bounded by
$2^{\card{E}}.2^{k.h_0([G])}$.
\end{proposition}


\subsection{The main lemma}

This definition formalises a notion that appears more or less clearly in the work of 
Steele and Yao, and of Ben-Or, as well as in the proof by Mulmuley. The vertices 
for paths of length $k$ in the $k$-th co-tree corresponds to the $k$-th cell 
decomposition, and the corresponding path defines the polynomials describing the
semi-algebraic set decided by a computational tree. While in Steele and Yao and 
Mulmuley's proofs, one obtain directly a polynomial for each cell, we here need to
construct a system of equations for each branch of the co-tree. 
%

%
%
%

Given a $\crew^{p}(\amcrealfull)$-graphing representative $G$ we will write $\rootdegree{G}$ 
the maximal value of $n$ for which an instruction $\bosqrtn{i}{j}$ appears in the realiser of an 
edge of $G$. 

\begin{lemma}\label{thm:graphingsBenOrsystems}\label{mainlemma}
Let $G$ be a computational graphing representative with edges realised only by 
generators of the \amc $\crew^{p}(\amcrealfull)$, and $\seqedges{k}{E}$ the 
set of length $k$ sequences of edges in $G$. Suppose $G$ computes the 
membership problem for $W \subseteq 
\realN^n$ in $k$ steps, i.e. for each element of $\realN^n$, 
$\pi_{\Space{S}}(G^{k}(x))=\top$ if and only if $x\in W$. Then $W$ is a semi-algebraic 
set defined by at most $\card{\seqedges{k}{E}}.2^{k.h_0([G])}$ systems of $pk$ 
equations of degree at most $\max(2,\rootdegree{G})$ and involving at most $pk+n$ variables.
\end{lemma}

The proof of this theorem is long but simple to understand. We define, for each
vertex of the $k$-th entropic co-tree, a system of algebraic equations (each of 
degree at most 2). The system is defined by induction on $k$, and uses the
information of the specific instruction used to extend the sequence indexing 
the vertex at each step. For instance, the case of division follows Ben-Or's 
method, introducing a fresh variable and writing down two equations as 
explained in Section \ref{subsec:cotrees}. 

%
%
%

\section{Recovering results from the literature}

\subsection{Ben-Or's theorem}

We now recover Ben-Or result by obtaining a bound on the number of
connected components of the subsets $W \subseteq \realN^n$ whose
membership problem is computed by a graphing in less than a given number
of iterations. This theorem is obtained by applying the Milnor-Thom theorem
on the obtained systems of equations to bound the number of connected
components of each cell. Notice that in this case $p=1$ and $\rootdegree{G}=2$ 
since the model of algebraic
computation trees use only square roots. A mode general result 
holds for algebraic computation trees extended with arbitrary roots, but we
here limit ourselves here to the original model.

\begin{theorem}\label{thm:graphingsBenOr}
Let $G$ be a computational $\amcrealfull$-graphing representative translating
an algebraic computational tree, $\seqedges{k}{E}$ the 
set of length $k$ sequences of edges in $G$. 
Suppose $G$ computes the membership 
problem for $W \subseteq \realN^n$ in $k$ steps. Then $W$ 
has at most $\card{\seqedges{k}{E}}.2^{k.h_0([G])+1}3^{2k+n-1}$ connected 
components.
\end{theorem}

Since a subset computed by a tree $T$ of depth $k$ is computed by
$\Interpret{T}$ in $k$ steps by \Cref{thm:quantsoundact}, we get as a
corollary the original theorem by Ben-Or 
relating the number of connected components of a set $W$ and the depth of 
the algebraic computational trees that compute the membership problem for $W$.

\begin{corollary}[{\cite[Theorem 5]{Ben-Or83}}]\label{ben-or}
  Let $W \subseteq \realN^n$ be any set, and let $N$ be the maximum of the
  number of connected components of $W$ and $\realN^n \setminus W$.
  An algebraic computation tree computing the membership problem for $W$ has
  height $\Omega(\log N)$.
\end{corollary}

%
%
%
%
%

\begin{remark}
In the case of algebraic \prams discussed in the next sections, the $k$-th entropic
co-tree $\cotree{k}[M]$ of a machine $M$ defines an algebraic computation tree 
which follows the $k$-th first steps of computation of $M$. I.e. the algebraic 
computation tree $\cotree{k}[M]$ approximate the computation of $M$ in such a way 
that $M$ and $\cotree{k}[M]$ behave in the exact same manner in the first $k$ steps.
\end{remark}

\subsection{Cucker's theorem}

Cucker's proof considers the problem defined as the following algebraic set.

\begin{definition}
Define $\cuckersproblem$ to be the set:
\[\{x\in \realN^\omega \mid \abs{x}=n \Rightarrow x_1^{2^n}+x_2^{2^n}=1 \},\]
where $\abs{x}=\max\{n\in\omega\mid x_n\neq 0\}$.
\end{definition}

It can be shown to lie within $\PtimeReal$, i.e. it is decided by a 
real Turing machine \cite{Blum:1989} -- i.e. working with real numbers and real 
operations --, running in polynomial time.

\begin{theorem}[Cucker (\cite{Cucker92}, Proposition 3)]
The problem $\cuckersproblem$ belongs to $\PtimeReal$.
\end{theorem}

We now prove that $\cuckersproblem$ is not computable by an algebraic 
circuit of polylogarithmic depth. The proof follows Cucker's argument, but 
uses the lemma proved in the previous section.

\begin{theorem}[Cucker (\cite{Cucker92}, Theorem 3.2)]
No algebraic circuit of depth $k=\log^i n$ and size $kp$ compute 
$\cuckersproblem$.
\end{theorem}

\begin{proof}
For this, we will use the lower bounds result 
obtained in the previous section. Indeed, by \Cref{thm:quantsoundalgcirc} and 
\Cref{mainlemma}, any problem decided 
by an algebraic circuit of depth $k$ is a semi-algebraic set defined by at most
$\card{\seqedges{k}{E}}.2^{k.h_0([G])}$ systems of $k$ 
equations of degree at most $\max(2,\rootdegree{G})=2$ (since only square roots
are allowed in the model) and involving at most 
$k+n$ variables. But the curve $\mathfrak{F}_{2^n}^{\realN}$ defined as 
$\{x_1^{2^n}+x_2^{2^n}-1=0\mid x_1,x_2\in\realN\}$ is infinite. As a consequence,
one of the systems of equations must describe a set containing an infinite number 
of points of $\mathfrak{F}_{2^n}^{\realN}$.

This set $S$ is characterized, up to some transformations on the set of
equations obtained from the entropic co-tree, by a finite system of
in\pointmedian{}equalities of the form
\[ \bigwedge_{i=1}^{s} F_i(X_1,X_2)= 0\wedge \bigwedge_{j=1}^{t} G_j(X_1, X_2) < 0, \]
where $t$ is bounded by $kp$ and the degree of the polynomials 
$F_i$ and $G_i$ are bounded by $2^k$. Moreover, since $\mathfrak{F}_{2^n}^{\realN}$ 
is a curve and no points in $S$ must lie outside of it, we must have $s>0$.

Finally, the polynomials $F_i$ vanish on that infinite subset of the curve and thus in a 
1-dimensional component of the curve. Since the curve is an irreducible one, this 
implies that every $F_i$ must vanish on the whole curve. Using the fact that the ideal 
$(X_1^{2^n} + X_2^{2^n}- 1)$ is prime (and thus radical), we conclude that all the $F_i$ 
are multiples of $X_1^{2^n} + X_2^{2^n}- 1$ which is impossible if their degree is 
bounded by $2^{\log^i n}$ as it is strictly smaller than $2^n$.
\end{proof}

%% file: surfaces-optimization.tex
\section{A proof that $\NCInteger\neq\Ptime$}

In this section, we provide a new presentation of a result of Mulmuley
which is part of his lower bounds for \enquote{prams without bit operations}. 
The idea is to encode a specific decision problem and the run of a \pram as two specific
subsets of the same space and show that no short run of the machine can define
the set of all instances of the decision problem. More specifically, consider
the problem \maxflow: given a weighted graph, find the maximal flow from a
source edge to a target edge. This is an optimization problem. It can be turned
into a decision problem by adding a new variable \(z\)—a threshold—and asking
whether there exists a solution greater than \(z\). This decision problem is known 
to be \Ptime-complete \cite{MaxflowComplete}.



\subsection{Geometric Interpretation of Optimization Problems}
\label{subsec:optprob}


Let $\OptProb$ be an optimization problem on $\realN^d$. Solving $\OptProb$ on
an instance $t$ amounts to optimizing a function $f_t(\cdot)$ over a space of
parameters. We note $\MaxOptProb(t)$ this optimal value. An affine function
$\Parametrization : [p;q] \to \realN^d$ is called a \emph{parametrization} of
$\OptProb$. Such a parametrization defines naturally a decision problem
$\DecProb$: for all $(x,y,z) \in \integerN^3$, $(x,y,z) \in \DecProb$ iff
$ z >0$, $x/z \in [p;q]$ and $y/z \leq \MaxOptProb\circ \Parametrization(x/z)$.

In order to study the geometry of $\DecProb$ in a way that makes its connection
with $\OptProb$ clear, we consider the ambient space to be $\realN^3$, and we
define the \emph{ray} $[p]$ of a point $p$ as the half-line starting at the
origin and containing $p$. The projection $\projectionAz{p}$ of a point $p$ on a
plane is the intersection of $[p]$ and the affine plane $\AffinePlane$ of
equation $z=1$. For any point $p \in \AffinePlane$, and all $ p_1 \in [p]$,
$\projectionAz{p_1} = p$. It is clear that for
$(p,p',q) \in \integerN^2\times \naturalN^+$,
$\projectionAz{(p,p',q)} = (p/q,p'/q,1)$.

The \emph{cone} $[C]$ of a curve $C$ is the set of rays of points of the
curve. The projection $\projectionAz{C}$ of a surface or a curve $C$ is the set
of projections of points in $C$. We note $\Frontier$ the frontier set
\( \Frontier = \{(x,y,1) \in \realN^3 \mid y = \MaxOptProb\circ
\Parametrization(x) \}\).  and we remark that
\( [\Frontier] = \{(x,y,z) \in \realN^2 \times \realN^+ \mid y/z = \MaxOptProb
\circ\Parametrization (x/z)\}.  \)

A machine $M$ decides the problem $\DecProb$ in $k$ steps if the partition 
of accepting cells in $\integerN^3$ induced by the machine -- i.e. the $k$-th 
cell decomposition -- is finer than the one defined by the problem's frontier 
$[\Frontier]$ (which is defined by the equation
$y/z \leq \MaxOptProb \circ \Parametrization(x/z)$).\newline

\noindent\textbf{Parametric Complexity.}
%
We now further restrict the class of problems we are interested in: we will only
consider $\OptProb$ such that $\Frontier$ is simple enough.
\begin{definition}
  We say that $\Parametrization$ is an \emph{affine parametrization} of
  $\OptProb$ if $\MaxOptProb \circ \Parametrization$ is convex, piecewise linear
  with breakpoints $\lambda_1 < \cdots < \lambda_{\rho}$, and such that all
  $(\lambda_i)_i$ and $(\MaxOptProb \circ \Parametrization(\lambda_i))_i$
  are rational.
  The \emph{parametric complexity} $\rho(\Parametrization)$ is the number of
  breakpoints $\rho$.
  The \emph{bitsize} of the parametrization is the maximum of the bitsizes of
  the numerators and denominators of the coordinates of the breakpoints of
  $\MaxOptProb \circ \Parametrization$.
\end{definition}

An optimization problem admitting an affine parametrization of complexity
$\rho$ is thus represented by a quite simple surface $[\Frontier]$: the
cone of the graph of a piecewise affine function, constituted of $\rho$
segments. We call such a surface is a \emph{$\rho$-fan} and define its
bitsize as \(\beta\) if all its breakpoints are rational and the bitsize of
their coordinates is less than $\beta$.

The restriction to such optimization problems seems quite dramatic when 
understood geometrically. Nonetheless, \maxflow admits such a parametrization.
  

\begin{theorem}[Murty \cite{murty1980computational}, Carstensen
  \cite{Carstensen:1983}]
  \label{thm:param}
  \label{maxflow-param} There exists an affine parametrization of bitsize
  $O(n^2)$ and complexity $2^{\Omega(n)}$ of the \maxflow problem for directed
  and undirected networks, where $n$ is the number of nodes in the network.
\end{theorem}

\noindent\textbf{Surfaces and fans.}
An algebraic surface in $\realN^3$ is a surface defined by an equation of the
form $p(x,y,z)=0$ where $p$ is a polynomial. If $S$ is a set of surfaces $S_i$, each
defined by a polynomial $p_i$, the \emph{total degree} of $S$ is defined as the sum of
the degrees of polynomials $p_i$.

Let $K$ be a compact of $\realN^3$ delimited by algebraic surfaces and $S$ be a
finite set of algebraic surfaces of total degree $\delta$. We can assume that
$K$ is delimited by two affine planes of equation $z=\mu$ and
$z=2\mu_z$ and the cone of a rectangle
$\{ (x,y,1) \mid |x|, |y| \leqslant \mu_{x,y}\}$, by taking any such compact
containing $K$ and adding the surfaces bounding $K$ to $S$. $S$ defines a
partition of $K$ by considering maximal compact subspaces of $K$ whose
boundaries are included in surfaces of $S$. Such elements are called the
\emph{cells} of the decomposition associated to $S$.

\begin{definition}
  Let $K$ be a compact of $\realN^3$.  
  A finite set of surfaces $S$ on $K$ \emph{separates} a $\rho$-fan $\Fan$ on
  $K$ if the partition on $\integerN^3 \cap K$ induced by $S$ is finer than the
  one induced by $\Fan$.
\end{definition}

A major technical achievement of Mulmuley \cite{Mulmuley99} -- not explicitly 
stated -- was to prove the following theorem, of purely geometric nature. We 
refer to the long version of this work\footnote{For the purpose of double-blind
reviews, we do not provide an explicit reference for the moment.}
for a detailed 
proof of this result.

\begin{theorem}[Mulmuley]
  \label{thm:mulmuley-geometric}
  Let $S$ be a finite set of algebraic surfaces of total degree $\delta$.
  There exists a polynomial $P$ such that, for all $\rho > P(\delta)$, $S$ does
  not separate $\rho$-fans.
\end{theorem}



%% file: result.tex
\subsection{Strengthening Mulmuley's result}

We will now prove our strengthening of Mulmuley's lower bounds for \enquote{\prams
without bit operations} \cite{Mulmuley99}. For this, we will combine the results from
previous sections to establish the following result.

\setcounter{theorem}{1}
\begin{theorem}
  \label{cor:main-pram}
  Let $G$ be a deterministic graphing interpreting a \pram with
  $2^{O((\log N)^c)}$ processors, where $N$ is the length of the inputs and $c$
  any positive integer.
  
  Then $G$ does not decide \maxflow in $O((\log N)^c)$ steps.
\end{theorem}

So, let $M$ be an integer-valued \pram. We can associate to it a real-valued \pram
\(\tilde{M}\) such that $M$ and $\tilde{M}$ accept the same (integer) values, and the ratio 
between the running time of the two machines is a constant. Indeed:
\begin{proposition}\label{prop:constanttimeeuclidian}
  Euclidian division can be computed by a constant time real-valued \pram.
\end{proposition}
%


\begin{proof}[Proof of \Cref{cor:main-pram}]
Suppose now that $\Interpret{M}$ has a finite set of edges $E$. 
Then $\Interpret{\tilde{M}}$ has too has a finite set of edge of cardinality
\(O(\card{E})\). Since the running time of the initial \pram over integers is 
equal, up to a constant, to the computation time of the $\crew^p(\amcrealfull)$-program 
$\Interpret{\tilde{M}}$, we deduce that if $M$ computes \maxflow in $k$ steps, then 
$\Interpret{\tilde{M}}$ computes \maxflow in at most $Ck$ steps where $C$ is a 
fixed constant.

By \Cref{thm:graphingsBenOrsystems}, the problem decided by $\Interpret{\tilde{M}}$ 
in $Ck$ steps defines a system of equations separating the integral inputs accepted by 
$M$ from the ones rejected. I.e. if $M$ computes \maxflow in $Ck$ steps, then this system
of equations defines a set of algebraic surfaces that separate the $\rho$-fan defined by 
\maxflow. Moreover, this system of equation has a total degree bounded 
by \(Ck\max(2,\rootdegree{G})2p\times 2^{O(\card{E})}\times 2^{k.h_0(\Interpret{\tilde{M}})}\).

By \Cref{thm:param} and \Cref{thm:mulmuley-geometric}, there exists a polynomial $P$
such that a finite set of algebraic surfaces of total degree $\delta$ cannot separate the 
$2^{\Omega(n)}$-fan defined by \maxflow as long as $2^{\Omega(n)}>P(\delta)$. But
here the entropy of $G$ is $O(p)$, as the entropy of a product $f\times g$ satisfies 
$h(f\times g)\leqslant h(f)+h(g)$ \cite{entropyproduct}. Hence $\delta=O(2^{p}2^{k})$, contradicting the hypotheses 
that $p=2^{O((\log N)^c)}$ and $k=2^{O((\log N)^c)}$.
\end{proof}

This has \Cref{Theorem1} as a corollary, which shows that the class $\NCInteger$ does
not contain $\maxflow$, and hence is distinct from $\Ptime$. The question of how 
this class relates to \(\NC\) is open: indeed, while bit extractions cannot be performed 
in constant time by our machines (a consequence of Theorem 
\ref{thm:graphingsBenOrsystems}), they can be simulated in logarithmic time.

%% file: omitted.tex
\section{Omitted proofs}

\begin{proof}[Proof of \Cref{prop:fullyact}]
A computation tree defines an $\amcact$-graphing $[T]$, and the
natural $\amcact$-graphing representative obtained from the inductive definition of $[T]$
is clearly an $\amcact$-treeing because $T$ is a tree. That this treeing represents faithfully
the computational tree $T$ raises no difficulty.

Let us now show that the membership problem of a subset $W\subseteq\realN^n$
that can be decided by a computational $\amcact$-treeing is also decided by an algebraic
computation tree $T$. We prove the result by induction on the number of states of the
computational $\amcact$-treeing. The initial case is when $T$ the set of states is exactly 
$\{1,\top,\bot\}$ with the order defined by $1<\top$ and $1<\bot$ and no other relations. This 
computational $\amcact$-treeing has at most 2 edges, since it is deterministic and the source 
of each edge is a subset among 
$\realN^{\omega}$, 
$\realN^{\omega}_{k\geqslant 0}$, 
$\realN^{\omega}_{k\leqslant 0}$, 
$\realN^{\omega}_{k> 0}$, 
$\realN^{\omega}_{k< 0}$, 
$\realN^{\omega}_{k=0}$, and 
$\realN^{\omega}_{k\neq 0}$. 

We first treat the case when there is only one edge of
source $\realN^n$. An element $(x_1,\dots,x_n)\in\realN^n$ is decided by $T$ if the 
main representative $((x_1,\dots,x_n,0,\dots),1)$ is mapped to $\top$. Since there is
only one edge of source the whole space, either this edge maps into the state $\top$ 
and the decided subset $W$ is equal to $\realN^n$, or it maps into $\bot$ and the 
subset $W$ is empty. In both cases, there exists an algebraic computation tree
deciding $W$. For the purpose of the proof, we will however construct a specific 
algebraic computation tree, namely the one that first computes the right expression
and then accepts or rejects. I.e. if the only edge mapping into $\top$ (resp. $\bot$)
is realised by an element $m$ in the \amc of algebraic computation trees which can 
be written as a product of generators $g_1,\dots,g_k$, we construct the tree of height 
$k+1$ that performs (in that order) the operations corresponding to $g_1$, $g_2$, etc., 
and then answers "yes" (resp. "no").

Now, the case where there is one edge of source a strict subspace, e.g. 
$\realN^{\omega}_{k\geqslant 0}$ (all other cases are treated in a similar manner)
and mapping into $\top$ (the other case is treated by symmetry). First, let us remark 
that if there is no other edge, one could very well add an edge to $T$ mapping into 
$\bot$ and realised by the identity with source the complementary subspace 
$\realN^{\omega}_{k< 0}$. We build a tree as follows. First, we test whether the 
variable $x_k$ is greater or equal to zero; this node has two children corresponding to
whether the answer to the test is "yes" or "no". We now construct the two subtrees
corresponding to these two children. The branch corresponding to "yes" is described
by the edge of source $\realN^{\omega}_{k\geqslant 0}$: we construct the tree of 
height $k+1$ performing the operations corresponding to the generators $g_1$, 
$g_2$, etc. whose product defined the realiser $m$ of $e$, and then answers "yes"
(resp. "no") if the edge $e$ maps into the state $\top$ (resp. $\bot$). Similarly, the 
other subtree is described by the realiser of the edge of source 
$\realN^{\omega}_{k< 0}$.

The result then follows by induction, plugging small subtrees as described above 
in place of the leaves of smaller subtrees.
\end{proof}

\begin{proof}[Proof of \Cref{lem:admseqHk}]
We show that the set $\admss{k}$ of admissible sequences of states of length $k$ has the same cardinality as the smallest subcover of $\mathcal{S}\vee [G]^{-1}(\mathcal{S})\vee\dots\vee [G]^{-(k-1)}(\mathcal{S}))$. Hence $H^{k}(G,\mathcal{S})=\frac{1}{k}\log_{2}(\card{\admss{k}})$, which implies the result.

The proof is done by induction. As a base case, we consider the set of $\admss{2}$ of  length $2$ admissible sequences of states and the cover $\mathcal{V}=\mathcal{S}\vee [G]^{-1}(\mathcal{S})$ of $D=[G]^{-1}(\Space{X})$. An element of $\mathcal{V}$ is an intersection $\Space{X}\times\{s_{1}\}\cap [G]^{-1}(\Space{X}\times\{s_{2}\})$, and is therefore equal to $C[s_{1},s_{2}]\times\{s_{1}\}$ where $C[s_{1},s_{2}]\subset\Space{X}$ is the set $\{x\in\Space{X}~|~[G](x,s_{1})\in\Space{X}\times\{s_{2}\}\}$. This set is empty if and only if the sequence $s_{1}s_{2}$ belongs to $\admss{2}$. Moreover, given another sequence of states $s'_{1}s'_{2}$, the sets $C[s_{1},s_{2}]$ and $C[s_{1},s_{2}]$ are disjoint. Hence a set $C[s_{1},s_{2}]$ is \emph{removable from the cover $\mathcal{V}$} if and only if $s_{1}s_{2}$ is not admissible. This proves the case $k=2$.

The step for the induction is similar. One considers the partition $\mathcal{S}_{k}=\bigvee_{i=0}^{-(k-1)}[G]^{i}(\mathcal{S})$ as $\mathcal{S}_{k-1}\vee[G]^{-(k-1)}(\mathcal{S})$. By the same argument, one shows elements of $\mathcal{S}_{k-1}\vee[G]^{-(k-1)}(\mathcal{S})$ are of the form $C[\mathbf{s}=(s_{0}s_{1}\dots s_{k-1}),s_{k}]\times\{s_{1}\}$ where $C[\mathbf{s},s_{k}]$ is the set $\{x\in\Space{X}~|~\forall i=2,\dots,k, [G]^{i-1}(x,s_{1})\in \Space{X}\times\{s_{i}\}\}$. Again, these sets $C[\mathbf{s},s_{k}]$ are pairwise disjoint and empty if and only if the sequence $s_{0}s_{1}\dots s_{k-1},s_{k}$ is not admissible.
\end{proof}

\begin{proof}[Proof of \Cref{lemma:coverentropybound}]
Let us fix some $\epsilon>0$. Notice that if we let $H_{k}(G,\mathcal{U})=H^{0}(\mathcal{U}\vee [G]^{-1}(\mathcal{U})\vee\dots\vee [G]^{-(k-1)}(\mathcal{U})))$, the sequence $H_{k}(U)$ satisfies $H_{k+l}
(\mathcal{U})\leqslant H_{k}(\mathcal{U})+H_{l}(\mathcal{U})$. By Fekete's lemma on subadditive sequences, this implies that $\lim_{k\rightarrow\infty}H_{k}/k$ exists and is equal to $\inf_{k}H_{k}/k$. Thus $h([G],\mathcal{U})=\inf_{k}H_{k}/k$. 

Now, the entropy $h([G])$ is defined as $\sup_{\mathcal{U}} \lim_{k\rightarrow\infty} H_{k}(\mathcal{U})/k$. This then rewrites as $\sup_{\mathcal{U}} \inf_{k} H_{k}(\mathcal{U})/k$. We can conclude that $h([G])\geqslant \inf_{k} H_{k}(\mathcal{U})/k$ for all finite open cover $\mathcal{U}$. 

Since $\inf_{k} H_{k}(\mathcal{U})/k$ is the limit of the sequence $H_{k}/k$, there exists an integer $N$ such that for all $k\geqslant N$ the following inequality holds: $\abs{H_{k}(\mathcal{U})/k-\inf_{k}H_{k}(\mathcal{U})/k}<\epsilon$, which rewrites as  $H_{k}(\mathcal{U})/k-\inf_{k}H_{k}(\mathcal{U})/k<\epsilon$. From this we deduce $H_{k}(\mathcal{U})/k<h([G])+\epsilon$, hence $H^{k}(G,\mathcal{U})<h([G])+\epsilon$ since $H^{k}(G,\mathcal{U})=H_{k}(G,\mathcal{U})$.
\end{proof}

\begin{proof}[Proof of \Cref{prop:entropy-k-cell-h0-edges}]
For a fixed sequence $\vec{e}$, the number of elements $n_{\vec{e}}^{\pi}$ 
of length $m$ in $\cotree{k}$ is bounded by the number of elements in the $m$-th 
cell decomposition of $T$, and is therefore bounded by $g(m)=2^{m.h_0([T])}$ 
by \autoref{prop:entropy-k-cell-h0}. 
The number of sequences $\vec{e}$ is bounded by $\card{\seqedges{k}{E}}$ and therefore 
the size of $\cotree{k}$ is thus bounded by $\card{\seqedges{k}{E}}.2^{(k+1).h_0([T])}$.
\end{proof}


\begin{proof}[Proof of \Cref{thm:graphingsBenOrsystems}]
If $G$ computes the membership problem for $W$ in $k$ steps, it means $W$
can be described as the union of the subspaces corresponding to the nodes 
$n^{\pi}_{\seq{e}}$ with $\pi$ of length $k$ in $\cotree{k}$. Now, each such 
subspace is an algebraic set, as it can be described by a set of polynomials as 
follows.

Finally let us note that, as in Mulmuley's work \cite{Mulmuley99}, since in our 
model the memory pointers are allowed to depend only on the nonnumeric 
parameters, indirect memory instructions can be treated as standard -- direct --
memory instructions. In other words, whenever an instruction involving a memory 
pointer is encountered during the course of execution, the value of the pointer 
is completely determined by nonnumerical data, and the index of the involved 
registers is completely determined, independently of the numerical inputs.


We define a system of equations $(E^{\seq{e}}_i)_{i}$ for each node 
$n^{\pi}_{\seq{e}}$ of the entropic co-tree $\cotree{k}$. We explicit the construction
for the case $p=1$, i.e. for the \amc $\crew^1(\amcrealfull)=\amcrealfull$; the case 
for arbitrary $p$ is then dealt with by following the construction and introducing
$p$ equations at each step (one for each of the $p$ instructions in $\amcrealfull$
corresponding to an element of $\crew^p(\amcrealfull)$). This is done 
inductively on the size of the path $\vec{e}$, keeping track of the last modifications
of each register. I.e. we define both the system of equations $(E^{\seq{e}}_i)_{i}$
and a function $\history{\seq{e}}: \realN^{\omega}+\bot\rightarrow \omega$ (which 
is almost everywhere null)\footnote{The use of $\bot$ is to allow for the creation of
fresh variables not related to a register.}. For an empty sequence, the system of
equations is empty, and the function $\history{\epsilon}$ is constant, equal to $0$.

Suppose now that $\vec{e'} = (e_1, \dots, e_m, e_m+1)$, with $\vec{e}=(e_1, 
\dots, e_m)$, and that one already computed $(E^{\seq{e}}_i)_{i\geqslant m}$ 
and the function $\history{\seq{e}}$. We now consider the edge $e_{m+1}$ and 
let $(r,r')$ be its realizer. We extend the system of equations 
$(E^{\seq{e}}_i)_{i\geqslant m}$ by a new equation $E_{m+1}$ and define the 
function $\history{\seq{e'}}$ as follows:
  \begin{itemize}
  \item if $r = \boadd{i}{j}{k}$, $\history{\seq{e'}}(x)=\history{\seq{e}}(x)+1$ if 
  	$x=i$, and $\history{\seq{e'}}(x)=\history{\seq{e}}(x)$ otherwise; then
	$E_{m+1}$ is $\benorvar{i} = \benorvar{j}+\benorvar{k}$;
  \item if $r = \bosubstract{i}{j}{k}$, $\history{\seq{e'}}(x)=\history{\seq{e}}(x)+1$ if 
  	$x=i$, and $\history{\seq{e'}}(x)=\history{\seq{e}}(x)$ otherwise; then
	$E_{m+1}$ is $\benorvar{i} = \benorvar{j} - \benorvar{k}$;
  \item if $r = \bomultiply{i}{j}{k}$, $\history{\seq{e'}}(x)=\history{\seq{e}}(x)+1$ if 
  	$x=i$, and $\history{\seq{e'}}(x)=\history{\seq{e}}(x)$ otherwise; then
	$E_{m+1}$ is $\benorvar{i} = \benorvar{j}\times \benorvar{k}$;
  \item if $r = \bodivide{i}{j}{k}$, $\history{\seq{e'}}(x)=\history{\seq{e}}(x)+1$ if 
  	$x=i$, and $\history{\seq{e'}}(x)=\history{\seq{e}}(x)$ otherwise; then
	$E_{m+1}$ is $\benorvar{i} = \benorvar{j}/\benorvar{k}$;
  \item if $r = \boaddconst{i}{k}{c}$, $\history{\seq{e'}}(x)=\history{\seq{e}}(x)+1$ if 
  	$x=i$, and $\history{\seq{e'}}(x)=\history{\seq{e}}(x)$ otherwise; then
	$E_{m+1}$ is $\benorvar{i} = c + \benorvar{k}$;
  \item if $r = \bosubstractconst{i}{k}{c}$, $\history{\seq{e'}}(x)=\history{\seq{e}}(x)+1$ if 
  	$x=i$, and $\history{\seq{e'}}(x)=\history{\seq{e}}(x)$ otherwise; then
	$E_{m+1}$ is $\benorvar{i} = c - \benorvar{k}$;
  \item if $r = \bomultiplyconst{i}{k}{c}$, $\history{\seq{e'}}(x)=\history{\seq{e}}(x)+1$ if 
  	$x=i$, and $\history{\seq{e'}}(x)=\history{\seq{e}}(x)$ otherwise; then
	$E_{m+1}$ is $\benorvar{i} = c\times \benorvar{k}$;
  \item if $r = \bodivideconst{i}{k}{c}$, $\history{\seq{e'}}(x)=\history{\seq{e}}(x)+1$ if 
  	$x=i$, and $\history{\seq{e'}}(x)=\history{\seq{e}}(x)$ otherwise; then
	$E_{m+1}$ is $\benorvar{i} = c/\benorvar{k}$;
  \item if $r = \bosqrtn{i}{k}$, $\history{\seq{e'}}(x)=\history{\seq{e}}(x)+1$ if 
  	$x=i$, and $\history{\seq{e'}}(x)=\history{\seq{e}}(x)$ otherwise; then
	$E_{m+1}$ is $\benorvar{i} = \sqrt[n]{\benorvar{k}}$;
  \item if $r = \identity$, the source of the edge $e_q$ is of the form 
    $\{ (x_1,\dots,x_{n+\ell}) \in \realN^{n+\ell} \mid P(x_k)\}\times \{i\}$ where
    $P$ compares the variable $x_k$ with $0$:
    \begin{itemize}
    \item if $P(x_k)$ is $x_k \neq 0$,  $\history{\seq{e'}}(x)=\history{\seq{e}}(x)+1$ if 
  	$x=\bot$, and $\history{\seq{e'}}(x)=\history{\seq{e}}(x)$ otherwise then 
	$E_{m+1}$ is $\benorvar{\bot}\benorvar{k} -1 = 0$;
    \item otherwise we set $\history{\seq{e'}}=\history{\seq{e}}$ and $E_{m+1}$ 
    equal to $P$.
  \end{itemize}
  \end{itemize}
  
  We now consider the system of equations $(E_{i})_{i=1}^{k}$ defined from
  the path $\seq{e}$ of length $k$ corresponding to a node $n^{\pi}_{\seq{e}}$
  of the $k$-th entropic co-tree of $G$. This system consists in $k$ equations 
  of degree at most $\max(2,\rootdegree{G})$ and containing at most $k+n$ variables, 
  counting the variables $x_1^0,\dots,x_n^0$ corresponding to the initial registers, 
  and adding at most $k$ additional variables since an edge of $\vec{e}$ introduces 
  at most one fresh variable. 
  Since the number of vertices $n^{\pi}_{\seq{e}}$ is bounded by 
  $\card{\seqedges{k}{E}}.2^{k.h_0([G])}$ by
  \autoref{prop:entropy-k-cell-h0-edges}, 
  we obtained the stated result in the case $p=1$.
  
  The case for arbitrary $p$ is then deduced by noticing that each step in the induction
  would introduce at most $p$ new equations and $p$ new variables. The resulting 
  system thus contains at most $pk$ equations of degree at most $\max(2,\rootdegree{G})$ 
  and containing at most $pk+n$ variables.
  %
%
\end{proof}

\begin{proof}[Proof of \Cref{thm:graphingsBenOr}]
By \Cref{thm:graphingsBenOrsystems} (using the fact that $p=1$ and 
$\rootdegree{G}=2$), the problem $W$ decided by $G$ in $k$ steps
is described by at most $\card{\seqedges{k}{E}}.2^{k.h_0([G])}$ systems of $k$ 
equations of degree $2$ involving at most $k+n$ variables. Applying 
\Cref{thm:milnor-ben-or}, we deduce that each such system of 
in\pointmedian{}equations (of $k$ equations of degree $2$ in $\realN^{k+n}$) 
describes a semi-algebraic variety $S$ such that
$\beta_0(S)<2.3^{(n+k)+k-1}$. This begin true for each of the 
$\card{\seqedges{k}{E}}.2^{k.h_0([G])}$ cells, we have that 
$\beta_0(W)<\card{\seqedges{k}{E}}.2^{k.h_0([G])+1}3^{2k+n-1}$.
\end{proof}

\begin{proof}[Proof of \Cref{ben-or}]
Let $T$ be an algebraic computation tree computing the membership problem 
for $W$, and consider the computational treeing $[T]$. Let $d$ be the height of
$T$; by definition of $[T]$ the membership problem for $W$ is computed in
exactly $d$ steps. Thus, by the previous theorem, $W$ has at most 
$\card{\seqedges{k}{E}}.2^{d.h_0([T])+1}3^{2d+n-1}$ connected components. 
As the interpretation of an algebraic computational tree, $h_0([T])$ is at most 
equal to $2$, and $\card{\seqedges{k}{E}}$ is bounded by $2^{d}$. Hence
$N\leqslant 2^d.2^{2d+1}3^{n-1}3^{2d}$, i.e. $d=\Omega(\log N)$.
\end{proof}

\begin{proof}[Proof of \Cref{prop:constanttimeeuclidian}]
  To compute \(p//q\), where \(p,q \in \integerN\), consider the real-valued
  machine such that the \(i^{\text{th}}\) processor computes \(x=p/q-i\) and if
  \(0 < x \leqslant 1\), writes \(i\) in the shared memory.
  This operation generalizes euclidian division and is computed in constant
  time. Moreover, this only uses a number of processor linear in the bitsize of
  the inputs if they are integers.
\end{proof}

%
%